\documentclass[a4paper,10pt]{article}

\usepackage{graphicx}
\usepackage[colorlinks=true,linkcolor=blue,urlcolor=blue]{hyperref}

\usepackage{tikz}
\tikzstyle{axis}=[]
\tikzstyle{function}=[thick]
\tikzset{font=\footnotesize}
\newcommand{\tikzwidth}{0.75}
\newcommand{\tikzheight}{0.3}
\newcommand{\tabfigheight}{0.18\textheight}

\usepackage{enumitem}
\setenumerate[1]{label=(\alph*)}
\setenumerate[2]{label=(\roman*)}

\usepackage{pdfsync}

\usepackage{filecontents}

\usepackage{needspace}

\usepackage{amsfonts}
\usepackage{amsmath}
\usepackage{amsthm}

\newtheorem{theorem}{Theorem}[section]
\newtheorem{proposition}[theorem]{Proposition}
\newtheorem{lemma}[theorem]{Lemma}

\theoremstyle{definition}
\newtheorem{definition}[theorem]{Definition}
\newtheorem{example}[theorem]{Example}
\newtheorem{construction}[theorem]{Construction}

\theoremstyle{remark}
\newtheorem{remark}[theorem]{Remark}

\numberwithin{equation}{section}
\numberwithin{theorem}{section}

\usepackage{dcolumn}
\newcolumntype{d}{@{}D{.}{.}{3}@{}}

\DeclareMathOperator{\successors}{succ}
\DeclareMathOperator{\dom}{dom}
\DeclareMathOperator{\epi}{epi}
\DeclareMathOperator{\conv}{conv}
\DeclareMathOperator{\cone}{cone}
\DeclareMathOperator{\cl}{cl}

\begin{document}

\title{American and Bermudan options in currency markets with proportional transaction costs}

\author{Alet Roux\thanks{Department of Mathematics, University of York, Heslington, YO10 5DD, United Kingdom. Email: alet.roux@york.ac.uk} \and Tomasz Zastawniak\thanks{Department of Mathematics, University of York, Heslington, YO10 5DD, United Kingdom. Email: tomasz.zastawniak@york.ac.uk}}

\date{}

\maketitle

\begin{abstract}
The pricing and hedging of a general class of options (including American, Bermudan and European options) on multiple assets are studied in the context of currency markets where trading is subject to proportional transaction costs, and where the existence of a risk-free num\'eraire is not assumed. Constructions leading to algorithms for computing the prices, optimal hedging strategies and stopping times are presented for both long and short option positions in this setting, together with probabilistic (martingale) representations for the option prices.


\end{abstract}

\section{Introduction}

{We} consider the pricing and hedging of a wide class of options within the model of foreign exchange markets proposed by {Kabanov \cite{kabanov1999}}, where proportional transaction costs are modelled as bid-ask spreads between currencies. This model has been well studied; see e.g.~{\cite{kabanov_rasonyi_stricker2002,kabanov_stricker2001b,schachermayer2004}}.

The results of this paper apply to any option that can be described in full by a payoff process together with an exercise policy specifying the circumstances in which it can be exercised at each date up to its expiration. {This includes} American, Bermudan and European options. For such options, we compute the ask price (seller's price, upper hedging price) as well as the bid price (buyer's price, lower hedging price), and derive probabilistic {(martingale)} representations for these prices. We also construct optimal superhedging trading strategies for the buyer and the seller, together with optimal stopping times consistent with the exercise {policy}.

{American options are being traded and hedged in large volumes throughout financial markets where transaction costs in the form of bid ask spreads are commonplace. The theory of American options under transaction costs developed to-date does not fully address the practical significance of the pricing and hedging problem in that it offers non-constructive existence proofs and tackles the short (seller's) position in American options only. This paper goes some way towards bridging the gap between the known theoretical results cited later in this introduction and practical considerations of being able to compute some prices, hedging strategies and stopping times for American options under transaction costs. In doing so this paper also provides alternative constructive proofs of the known results, and extends these results to both parties to the option contract, that is, not just the holder of a {short} position (the seller) but also to the party holding a long position (the buyer)
in the option.}

It is well known in complete models without transaction costs that the best stopping time for the holder of an American or Bermudan option is also the most expensive stopping time for the seller to hedge against, and that hedging against this particular stopping time protects the seller against all other stopping times. {Chalasani and Jha \cite{chalasani_jha2001}} observed that this is no longer the case for American options in the presence of proportional transaction costs: to hedge against all {(ordinary)} stopping times, the seller must in effect be protected against a certain \emph{randomised} stopping time {(see Definition~\ref{def:randomised-stopping-time})}. Thus the optimal stopping times of the buyer and seller of an American option no longer coincide, and it may cost the seller more to hedge against all stopping times than to hedge against the best stopping time for the buyer. {This} is true in general for any option that allows more than one exercise time (i.e.~any non-European option).

There is a geometrical explanation for this apparent lack of symmetry. For both parties to an option, the price, optimal stopping time and optimal superhedging strategy solve a linear optimization problem over the set of superhedging strategies. The superhedging strategies for the seller form a convex set. In contrast, each superhedging strategy for the buyer hedges against a specific stopping time, so that a convex combination of two superhedging strategies for different stopping times may no longer be a superhedging strategy for the buyer. Thus the pricing problem~\eqref{eq:pi-a} for the seller is convex, whereas if the exercise policy allows more than one stopping time, then the pricing problem~\eqref{eq:pi-b} for the buyer is a mixed integer programming problem that is generally not convex (not even in the friction-free case; for American options see {\cite{Pennanen_King2004}}).

The linear optimization problems~\eqref{eq:pi-a} and \eqref{eq:pi-b} both grow exponentially with the number of time steps, even for options with path-independent payoffs (see {\cite{Chen_Palmer_Sheu2008,rutkowski1998}} for results on European options). Various special cases of European and American options have been studied in binomial two-asset models with proportional transaction costs. The replication of European options has been well studied {(see e.g.~\cite{bensaid_lesne_pages_scheinkman1992,boyle_vorst1992,palmer2001a,perrakis_lefoll1997})}, and the first algorithm (with exponential running time) for computing the bid and ask prices for European options was established by Edirisinghe, Naik and Uppal \cite{edirisinghe_naik_uppal1993}. In a similar technical setting, {Koci\'nski \cite{kocinski1999,kocinski2001b}} studied the exact replication of American options, {Perrakis and Lefoll \cite{perrakis_lefoll2000,perrakis_lefoll2004}} investigated the pricing of American call and put options, and {Tokarz and Zastawniak
\cite{tokarz_zastawniak2006}} worked with general American options under small proportional transaction costs. Recently, {Loehne and Rudloff \cite{Loehne_Rudloff2011}} established {an algorithm} for finding the set of superhedging strategies for European {(but not for American or Bermudan)} options in a similar technical setting to the present paper.

{The main contribution of this paper is to provide constructive and efficient algorithms for computing the option prices, optimal
hedging strategies and stopping times for both the long and short positions in
American-style options in multi-asset markets under proportional transaction
costs. Another goal is to establish in a constructive manner probabilistic
(martingale) representations for American-style options for both the seller's
(long) and buyer's (short) positions in such options.

Previous work in this direction involves non-constructive representation
theorems for the short position in American options. This
includes the pioneering paper by {Chalasani and Jha \cite{chalasani_jha2001}}, who treated
American options with cash settlement and no transaction costs at the time of settlement in a single-stock market model.
Moreover, for American options in currency markets, {Bouchard and Temam \cite{bouchard_temam2005}}
established dual representations for the set of initial endowments that allow
to superhedge the short position. Their work, based on a non-constructive
existence argument, allows for a general setting based on an arbitrary
probability space. Similar work has been carried out in a continuous time model {\cite{Bouchard_Chassagneux2009,DeValliere_Denis_Kabanov2009}}.

The convex duality methods deployed in these papers do not, however, lend
themselves to studying hedging or pricing for the opposite party to an
American option contract, namely the option's buyer, as this involves an
inherently non-convex optimisation problem. Ideas going beyond convex duality
are necessary and are developed here.

The constructions and numerical  algorithms put
forward in the present paper call naturally for a discretisation. It is a
reasonable compromise between admitting models based or arbitrary probability spaces and possibly continuous time (such work involves topological and
functional analytic questions of theoretical interest, but non-constructive existence proofs) and being able to actually compute the prices, hedging strategies
and stopping times (as demanded by the applied nature of the problem in hand), and the dual counterparts thereof.

The constructive results for American-style options in multi-asset markets
under transaction costs are new. Similar questions were studied by {Loehne and Rudloff \cite{Loehne_Rudloff2011}} for European options, also in the discrete setting. Their
results on European options are covered by the present work as a special case.
In fact, even when specialised to European options, our results are still slightly more general as we are able to relax the
robust no-arbitrage condition of {Schachermayer \cite{schachermayer2004}} that was assumed {in
\cite{Loehne_Rudloff2011}}, and require just the weak no-arbitrage property
(2.4) of {Kabanov and Stricker \cite{kabanov_stricker2001b}}.}

The {proofs} of the main results (Theorems \ref{th:seller} and \ref{th:buyer}) {include} constructions of the sets of superhedging strategies and stopping times for both the buyer and seller, together with the approximate martingales and pricing measures involved in the martingale representations of both the bid and ask price of an option with general exercise policy (subject to mild regularity conditions) on multiple assets under proportional transaction {costs in} a general discrete time setting. Such constructions extend and improve upon each of the various special cases mentioned above, as well as the results we previously reported for European and American options in two-asset models \cite{Roux_Tokarz_Zastawniak2008,Roux_Zastawniak2009}. These constructions are efficient in that their running length grows only polynomially with the number of time steps when pricing options with {path-dependent} payoffs and exercise policies in recombinant tree models.

The paper is organised as follows. In Section~\ref{sec:preliminaries} we fix the notation, specify the market model with transaction costs, and review {various notions concerning convex sets and functions}, randomised stopping times and approximate martingales. The notion of an exercise policy is introduced in Section~\ref{sec:exercise-policy}. The main pricing and hedging results for the buyer and seller are presented in Section \ref{sec:main-results} as Theorems \ref{th:seller} and \ref{th:buyer}, and various special cases are discussed. Section~\ref{sec:seller} is devoted to the proof of Theorem \ref{th:seller} for the seller, while Theorem \ref{th:buyer} is proved in Section \ref{sec:buyer}. In Section \ref{sec:numerical} the constructions in Sections~\ref{sec:seller} {and}~\ref{sec:buyer} are applied to two realistic examples. Appendix~\ref{appendix} gives the proof of a technical lemma used in the proof of Theorem \ref{th:seller}.

\section{Preliminaries and notation}
\label{sec:preliminaries}

\subsection{Convex sets and functions}

For any set $A\subseteq\mathbb{R}^d$, define
\[
 \sigma_i(A) := \{x=(x^1,\ldots,x^d)\in A : x^i = 1\},
\]
and define the \emph{cone generated by $A$} as
\[
 \cone A := \{\lambda x:\lambda\ge0,x\in A\}.
\]
We say that a non-empty cone $C\subseteq\mathbb{R}^d$ is \emph{compactly $i$-generated} if $\sigma_i(C)$ is compact, non-empty and $C=\cone \sigma_i(C)$.

Let $\cdot$ denote the scalar product in {$\mathbb{R}^d$}. For any {non-empty convex cone} $A\subseteq\mathbb{R}^d$, denote by $A^\ast$ the polar of $-A$, i.e.
\[
 A^\ast := \{y\in\mathbb{R}^d:y\cdot x \ge0\text{ for all }x\in A\}.
\]
If $A$ is a non-empty closed convex cone, then $A^\ast$ is also a non-empty closed convex cone {\cite[Theorem 14.1]{rockafellar1996}}.

The \emph{effective domain} of any convex function $f:\mathbb{R}^d\rightarrow\mathbb{R}\cup\{+\infty,-\infty\}$ is defined as
\[
 \dom f := \{y\in\mathbb{R}^d : f(y)<\infty\}.
\]
The \emph{epigraph} of $f$ is defined as
\[
 \epi f := \{(y_0,y)\in\mathbb{R}\times\mathbb{R}^d:y_0\ge f(y)\}.
\]
The function $f$ is called \emph{proper} if $\epi f\neq\emptyset$ and $f(y)>-\infty$ for all $y\in\mathbb{R}^d$.

Define the \emph{convex hull} $\conv A$ of any set $A\subseteq\mathbb{R}^d$ as the smallest convex set containing $A$. Define the \emph{convex hull} of a finite collection $g_1,\ldots,g_n:\mathbb{R}^d\rightarrow\mathbb{R}\cup\{\infty\}$ of proper convex functions as the greatest convex function majorised by $g_1,\ldots,g_n$, equivalently
\[
 \conv\{g_1,\ldots,g_n\}(x) := \inf \sum_{k=1}^n\alpha_kg_k(x_k)
\]
for each $x\in\mathbb{R}^d$, where the infimum is taken over all $x_k\in\mathbb{R}^d$ and $\alpha_k\ge0$ for $k=1,\ldots,n$ such that
\begin{align*}
 \sum_{k=1}^n\alpha_k &= 1, & \sum_{k=1}^n\alpha_kx_k &= x.
\end{align*}
Also note that
\[
 \dom \conv\{g_1,\ldots,g_n\} = \conv\left[\bigcup_{k=1}^n\dom g_k\right].
\]

The \emph{closure} $\cl f$ of a proper convex function $f:\mathbb{R}^d\rightarrow\mathbb{R}\cup\{\infty\}$ is defined as the unique function whose epigraph is
\begin{equation} \label{eq:closure-of-function}
 \epi (\cl f) = \overline{\epi f}.
\end{equation}
If $f$ is not proper, then $\cl f$ is defined as the constant function $-\infty$. A proper convex function $f$ is called \emph{closed} if $f=\cl f$, equivalently if $\epi f$ is closed.

Define the \emph{support function} $\delta^\ast_A:\mathbb{R}^d\rightarrow\mathbb{R}\cup\{\infty\}$ of a non-empty convex set $A\subseteq\mathbb{R}^d$ as
\[
 \delta^\ast_A (x) := \sup\{x\cdot y: y\in A\}.
\]
The function $\delta^\ast_A$ is convex, proper and positively homogeneous. If $A$ is closed, then $\delta^\ast_A$ is closed {\cite[Theorem 13.2]{rockafellar1996}}.
We shall make use of the identity
\begin{equation} \label{eq:lem:support-of-cone:Rd}
 \delta^\ast_{\mathbb{R}^d}(y)=
\begin{cases}
 0 &\text{if } y=0,\\
 \infty &\text{if } y\neq0.
\end{cases}
\end{equation}

\subsection{Proportional transaction costs in a currency market model}

{We} consider a market model with~$d$ assets (henceforth referred to as currencies following the terminology of {Kabanov \cite{kabanov1999} and others}) and discrete trading dates~$t=0,1,\ldots,T$ on a finite probability space~$(\Omega,\mathcal{F},{\mathbb{P}})$ with filtration~$(\mathcal{F}_t)$. The exchange rates between the currencies are represented as an adapted matrix-valued process~$(\pi^{ij}_t)_{i,j=1}^d$, where for any~$t=0,\ldots,T$ and~$i,j=1,\ldots,d$ the quantity~$\pi_t^{ij}>0$ is the amount in currency~$i$ that needs to be exchanged in order to receive one unit of currency~$j$ at time~$t$.

 We assume without loss of generality that $\mathcal{F}_0$ is trivial, that $\mathcal{F}_T = 2^\Omega$ and that ${\mathbb{P}}(\omega)>0$ for all $\omega\in\Omega$. Let~\(\Omega_t\) be the collection of atoms (called \emph{nodes}) of~\(\mathcal{F}_t\) at any time~\(t\). A node~\(\nu\in\Omega_{t+1}\) at time~\(t+1\) is called a \emph{successor} of a node~\(\mu\in\Omega_t\) at time~\(t\) if~\(\nu\subseteq\mu\). Denote the collection of successors of any node~\(\mu\) by~\(\successors \mu\).

We write~$\mathcal{L}_t$ for the family of~$\mathcal{F}_t$-measurable~$\mathbb{R}^d$-valued random variables, where for convenience~\mbox{$\mathcal{L}_{-1}:=\mathcal{L}_0$}. Throughout this paper we shall implicitly and uniquely identify random variables in $\mathcal{L}_t$ with functions on $\Omega_t$, and we shall throughout adopt the notation
\[
 \mathcal{A}^\mu:=\{X^{\mu}:X\in\mathcal{A}\}\text{ for all }\mathcal{A}\subseteq\mathcal{L}_t, \mu\in\Omega_t.
\]

Writing~$\mathcal{L}_t^+$ for the family of non-negative random variables in~$\mathcal{L}_t$, a portfolio~$x=(x^1,\ldots,x^d)\in\mathcal{L}_t$ is called {\emph{solvent}} whenever it can be exchanged into a portfolio in~$\mathcal{L}_t^+$ without additional investment, i.e.\ if there exist~$\mathcal{F}_t$-measurable random variables~$\beta^{ij}\ge0$ for~$i,j=1,\ldots,d$ such that
\begin{equation}\label{eq:1.1}
	x^j + \sum_{i=1}^d\beta^{ij} - \sum_{i=1}^d\beta^{ji}\pi_t^{ji}\ge0 \text{ for all } j.
\end{equation}
Here~$\beta^{ij}$ represents the number of units of currency~$j$ obtained by exchanging currency~$i$. {The solvency condition~\eqref{eq:1.1} can be written as
\[
	x\in\mathcal{K}_t,
\]
where $\mathcal{K}_t$ is the convex cone in~$\mathcal{L}_t$ generated by the unit vectors $e^i$, $i=1,\ldots,d$ forming the canonical basis in $\mathbb{R}^d$ and the vectors $e^i\pi^{ij}_t-e^j$, $i,j=1,\ldots,d$. We refer to~$\mathcal{K}_t$ as the \emph{solvency cone}. Observe that $\mathcal{K}_t$ is a polyhedral cone and therefore closed.}

A \emph{self-financing strategy}~$y=(y_t)$ is a predictable $\mathbb{R}^d$-valued process with initial value $y_0\in\mathcal{L}_0=\mathbb{R}^d$ such that
\[
	y_t - y_{t+1}\in\mathcal{K}_t \text{ for all }t<T.
\]
Denote the set of all self-financing strategies by~$\Phi$.

The model with transaction costs is said to satisfy the \emph{weak no-arbitrage property} $(\mathrm{NA}^\mathrm{w})$ of {Kabanov and Stricker \cite{kabanov_stricker2001b}} if
\begin{equation}\label{eq:NA}
\left\{y_T : y\in\Phi \text{ and } y_0=0\right\} \cap \mathcal{L}^{+}_T=\left\{0\right\}.
\end{equation}
This formulation is formally different but equivalent to that of {\cite{kabanov_stricker2001b}}, and was introduced by {Schachermayer \cite{schachermayer2004}}, who called it simply the no-arbitrage property.

We have the following fundamental result.

\begin{theorem}[\cite{kabanov_stricker2001b,schachermayer2004}] \label{th:ftap}
	The model satisfies the weak no-arbitrage property if and only if there exist a probability measure~$\mathbb{Q}$ equivalent to~${\mathbb{P}}$ and an~$\mathbb{R}^d$-valued $\mathbb{Q}$-martingale $S=(S^1_t,\ldots,S^d_t)$ such that
	\begin{equation}\label{eq:1.2}
		0 < S^j_t\le\pi^{ij}_tS^i_t \text{ for all }i,j,t.
	\end{equation}
\end{theorem}

\begin{remark} \label{rem:St*-generation}
Condition~\eqref{eq:1.2} can equivalently be written as
\[
	S_t\in\mathcal{K}^\ast_t\setminus\{0\}\text{ for all }t.
\]
If the model satisfies the weak no-arbitrage property, then $\mathcal{K}^\ast_t$ is a non-empty polyhedral cone, and it is compactly $i$-generated with
\begin{multline*}
 \sigma_i(\mathcal{K}_t^{\ast\mu}) = \left\{(s^1,\ldots,s^d)\in\mathbb{R}^d:s^i=1,\tfrac{1}{\pi^{ji\mu}_t}\le s^j\le\pi^{ij\mu}_t\text{ for all }j\neq i,\right. \\
\left. s^j \le \pi^{kj\mu}_ts^k\text{ for all }j\neq i,k\neq i\right\}
\end{multline*}
for all $\mu\in\Omega_t$.
\end{remark}

\begin{definition}[Equivalent martingale pair]
A pair~$(\mathbb{Q},S)$ satisfying the conditions of Theorem~\ref{th:ftap} is called an \emph{equivalent martingale pair}.
\end{definition}

Denote the family of equivalent martingale pairs by $\mathcal{P}$. Let
\[
 \mathcal{P}^i:=\{(\mathbb{Q},S)\in\mathcal{P}:S^i_t=1\text{ for all }t\}
\]
for all $i=1,\ldots,d$.

We assume from here on that the model satisfies the weak no-arbitrage property, so that $\mathcal{P}\neq\emptyset$, equivalently $\mathcal{P}^i\neq\emptyset$ for all $i$.

\begin{example}
\label{Exl:1}\upshape Consider three assets, where asset 3 is a cash account.
Suppose that in a friction-free market assets~1 and~2 can be bought/sold,
respectively, for $S^{1}=12$ and $S^{2}=8$ units of cash (asset 3). The
friction-free exchange rate matrix would then be%
\[
\left[
\begin{array}
[c]{ccc}%
1 & S^{2}/S^{1} & 1/S^{1}\\
S^{1}/S^{2} & 1 & 1/S^{2}\\
S^{1} & S^{2} & 1
\end{array}
\right]  =\left[
\begin{array}
[c]{ccc}%
1 & 2/3 & 1/12\\
3/2 & 1 & 1/8\\
12 & 8 & 1
\end{array}
\right]  .
\]
Now assume that whenever an asset~$i$ is exchanged into a different asset~$j$,
transaction costs are charged at a fixed rate $k\geq0$ against asset~$i$,
resulting in each off-diagonal exchange rate increased by a factor $1+k$. If
$k=\frac{1}{3}$, the exchange rate matrix becomes%
\[
\pi=\left[
\begin{array}
[c]{ccc}%
1 & (1+k)S^{2}/S^{1} & (1+k)/S^{1}\\
(1+k)S^{1}/S^{2} & 1 & (1+k)/S^{2}\\
(1+k)S^{1} & (1+k)S^{2} & 1
\end{array}
\right]  =\left[
\begin{array}
[c]{ccc}%
1 & 8/9 & 1/9\\
1/4 & 1 & 1/6\\
16 & 32/3 & 1
\end{array}
\right]  .
\]
The cone $\mathcal{K}$ consisting of solvent portfolios $(x^{1},x^{2},x^{3})$
and the section $\sigma_{3}(\mathcal{K}^{\ast})$, which generates the
cone~$\mathcal{K}^{\ast}$, are shown in Figures~\ref{s_and_sstar.png}(a), (b),
respectively.
\begin{figure}
\includegraphics[
natheight=1.969200in,
natwidth=4.479700in,
height=1.6025in,
width=3.6123in
]%
{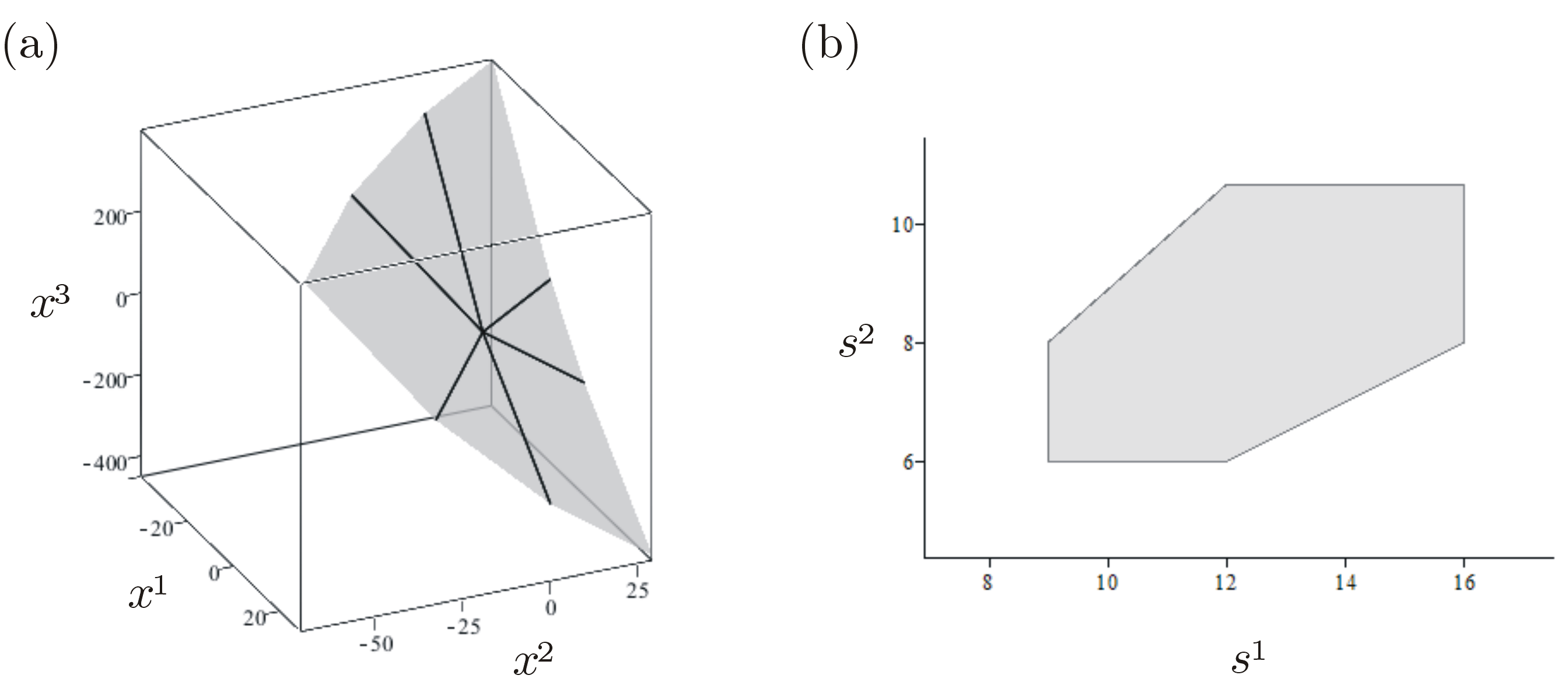}%
\caption{Solvency cone $\mathcal{K}$ and section $\sigma_{3}%
(\mathcal{K}^{\ast})$ of the cone $\mathcal{K}^{\ast}$, Example~\ref{Exl:1}}%
\label{s_and_sstar.png}%
\end{figure}
\end{example}

\subsection{Randomised stopping times}

\begin{definition}[Randomised stopping time] \label{def:randomised-stopping-time}
A \emph{randomised (or mixed) stopping time}~\(\chi=(\chi_t)\) is a non-negative adapted process such that
\[
	\sum_{t=0}^T\chi_t = 1.
\]
We write~$\mathcal{X}$ for the collection of all randomised stopping times.
\end{definition}

Let $\mathcal{T}$ be the set of (ordinary) stopping times. Any stopping time~$\tau\in\mathcal{T}$ can be {identified with} the randomised stopping time~$\chi^\tau=(\chi^\tau_t)\in\mathcal{X}$ defined by
\[
	\chi^\tau_t:=\mathbf{1}_{\{\tau=t\}}
\]
for all $t$. Here $\mathbf{1}_A$ denotes the indicator function of $A\subseteq\Omega$.

For any adapted process~$X=(X_t)$ and~$\chi\in\mathcal{X}$, define the \emph{value of~$X$ at~$\chi$} as
\[
	X_\chi := \sum_{t=0}^T\chi_tX_t.
\]
Moreover, define the processes~\(\chi^\ast=(\chi^\ast_t)\) and~\(X^{\chi\ast}_{\phantom t}=(X^{\chi\ast}_t)\) as
\begin{align*}
	\chi_t^\ast &:= \sum_{s=t}^T\chi_s, & X_t^{\chi\ast} &:= \sum_{s=t}^T\chi_sX_s
\end{align*}
for all~$t$. Observe that~\(\chi^\ast\) is a predictable process since
\[
	\chi_t^\ast = 1 - \sum_{s=0}^{t-1}\chi_s
\]
whenever~$t>0$. For notational convenience define
\begin{align}\label{eq:1.4}
	\chi_{T+1}^\ast &:= 0, & X_{T+1}^{\chi\ast} &:= 0.
\end{align}

\begin{definition}[{$\chi$}-approximate martingale pair] \label{def:chi-eqv-pricing-measure-and-pricing-process}
For any~$\chi\in\mathcal{X}$ the pair~$(\mathbb{Q},S)$ is called a~\emph{\(\chi\)-approximate martingale pair} if~$\mathbb{Q}$ is a probability measure and~$S$ an adapted process satisfying
\begin{align*}
	S_t&\in\mathcal{K}_t^\ast\setminus\{0\}, & \mathbb{E}_{\mathbb{Q}}(S_{t+1}^{\chi\ast}|\mathcal{F}_t) &\in\mathcal{K}_t^\ast
\end{align*}
for all~$t$. If in addition~$\mathbb{Q}$ is equivalent to~${\mathbb{P}}$, then~$(\mathbb{Q},S)$ is called a~\(\chi\)-approximate \emph{equivalent} martingale pair.
\end{definition}

Denote the family of \(\chi\)-approximate martingale pairs~$(\mathbb{Q},S)$ by~$\bar{\mathcal{P}}(\chi)$, and write $\mathcal{P}(\chi)$ for the family of~\(\chi\)-approximate equivalent martingale pairs. For an ordinary stopping time $\tau\in\mathcal{T}$ we write $\mathcal{P}(\tau):=\mathcal{P}(\chi^\tau)$ and $\bar{\mathcal{P}}(\tau):=\bar{\mathcal{P}}(\chi^\tau)$ and say that $(\mathbb{Q},S)$ is a $\tau$-approximate (equivalent) martingale pair whenever it is a $\chi^\tau$-approximate (equivalent) martingale pair.

For any $\chi\in\mathcal{X}$ and $i=1,\ldots,d$ define
\begin{align*}
 \bar{\mathcal{P}}^i(\chi)&:=\{(\mathbb{Q},S)\in\bar{\mathcal{P}}(\chi):S^i_t=1\text{ for all }t\},\\
\mathcal{P}^i(\chi)&:=\{(\mathbb{Q},S)\in\mathcal{P}(\chi):S^i_t=1\text{ for all }t\}.
\end{align*}
Noting that
$\mathcal{P} \subseteq \mathcal{P}(\chi) \subseteq \bar{\mathcal{P}}(\chi)$,
it follows that
$\mathcal{P}^i \subseteq \mathcal{P}^i(\chi) \subseteq \bar{\mathcal{P}}^i(\chi)$,
and the weak no-arbitrage property implies that all these families are non-empty.

We have the following simple result.

\begin{lemma} \label{lem:closure}
	Fix any $i=1,\ldots,d$, and let~$\xi$ be any adapted $\mathbb{R}^d$-valued process. Then for any~$\delta>0$, any~$\chi\in\mathcal{X}$ and any~$(\bar{\mathbb{Q}},\bar{S})\in\bar{\mathcal{P}}^i(\chi)$ there exists a $\chi$-approximate martingale pair $(\mathbb{Q}^\delta,S^\delta)\in\mathcal{P}^i(\chi)$ such that
	\[
		\lvert\mathbb{E}_{\mathbb{Q}^\delta}((\xi\cdot S^\delta)_\chi) - \mathbb{E}_{\bar{\mathbb{Q}}}((\xi\cdot \bar{S})_\chi)\rvert < \delta.
	\]
\end{lemma}
\begin{proof}
	The weak no-arbitrage property guarantees the existence of some~$(\mathbb{Q},S)\in\mathcal{P}^i\subseteq\mathcal{P}^i(\chi)$. If $\mathbb{E}_{\mathbb{Q}}((\xi\cdot S)_\chi)= \mathbb{E}_{\bar{\mathbb{Q}}}((\xi\cdot \bar{S})_\chi)$, then the claim holds with $\mathbb{Q}^\delta:=\mathbb{Q}$ and $S^\delta:=S$. If not, fix $\varepsilon := \min\left\{1,\frac{\delta}{2}/\lvert\mathbb{E}_{\mathbb{Q}}((\xi\cdot S)_\chi) - \mathbb{E}_{\bar{\mathbb{Q}}}((\xi\cdot \bar{S})_\chi)\rvert\right\}$, let $\mathbb{Q}^\delta := (1-\varepsilon)\bar{\mathbb{Q}} + \varepsilon\mathbb{Q}$ and
	\[
		S^\delta_t := (1-\varepsilon)\bar{S}_t\mathbb{E}_{\mathbb{Q}^\delta}\!\!\left(\left.\tfrac{d\bar{\mathbb{Q}}}{d\mathbb{Q}^\delta}\right|\mathcal{F}_t\right) + \varepsilon S_t \mathbb{E}_{\mathbb{Q}^\delta}\!\!\left(\left.\tfrac{d\mathbb{Q}}{d\mathbb{Q}^\delta}\right|\mathcal{F}_t\right) \text{ for all }t.
	\]
\end{proof}

\section{Exercise policies}
\label{sec:exercise-policy}

In the next section and onwards we will consider the pricing and hedging of an option that may only be exercised in certain situations, namely at any time~$t$ the owner of the option is only allowed to exercise on a subset $\mathcal{E}_t$ of $\Omega$. This setting contains a wide class of options, for example:
\begin{itemize}
 \item A European option corresponds to $\mathcal{E}_T=\Omega$ and $\mathcal{E}_t=\emptyset$ for all $t<T$.

 \item A Bermudan option with exercise dates $t_1<\ldots<t_n$ corresponds to
\begin{equation} \label{eq:bermudan-exercise-policy}
 \mathcal{E}_t =
\begin{cases}
 \Omega & \text{if }t=t_1,\ldots,t_n, \\
 \emptyset & \text{otherwise.}
\end{cases}
\end{equation}

 \item An American option corresponds to $\mathcal{E}_t=\Omega$ for all $t$.

 \item An American-style option with random expiration date $\tau\in\mathcal{T}$ corresponds to $\mathcal{E}_t=\{\tau\ge t\}$ for all $t$.
\end{itemize}

The introduction of exercise policies allows the unification {of} results for specific types of options, most notably European and American {ones}. More immediately, we {shall use} exercise policies as a theoretical tool in Section \ref{sec:buyer} when deriving the pricing and hedging theorem for the buyer {of an option with a general exercise policy} from {the} corresponding results for the seller of a related European-style option.

An exercise policy is formally defined as follows.

\begin{definition}[Exercise policy]
An \emph{exercise policy} $\mathcal{E}\equiv(\mathcal{E}_t)$ is a sequence of subsets of $\Omega$ such that $\mathcal{E}_t\in\mathcal{F}_t$ for all $t$,
\begin{equation} \label{eq:def:exercise-policy:1}
 \bigcup_{s=t+1}^T\mathcal{E}_s\in\mathcal{F}_t \text{ for all }t<T,
\end{equation}
and
  \begin{equation} \label{eq:def:exercise-policy:2}
   \bigcup_{t=0}^T\mathcal{E}_t = \Omega.
  \end{equation}
\end{definition}

The condition $\mathcal{E}_t\in\mathcal{F}_t$ is consistent with the intuitive notion of allowing the buyer to make exercise decisions based on information available at time $t$. Condition \eqref{eq:def:exercise-policy:1} is consistent with allowing the buyer to determine on the basis of information currently available whether or not there are future opportunities for exercise. Condition \eqref{eq:def:exercise-policy:2} ensures that there is at least one opportunity to exercise the option in each scenario.

Define the sequence $\mathcal{E}^\ast=(\mathcal{E}^\ast_t)$ of sets associated with an exercise policy $\mathcal{E}$ as
\[
 \mathcal{E}^\ast_t:=\bigcup_{s=t}^T\mathcal{E}_s
\]
for all $t$. For each $t$, the set $\mathcal{E}^\ast_t$ contains those scenarios in which it is possible to exercise the option in at least one of the time steps $t, \ldots,T$. Write $\mathcal{E}^\ast_{T+1}:=\emptyset$ for convenience.

For an exercise policy $\mathcal{E}=(\mathcal{E}_t)$, define the sets of \emph{randomised and ordinary stopping times consistent with $\mathcal{E}$} as
\begin{align*}
 \mathcal{X}^\mathcal{E} &:= \{\chi\in\mathcal{X} : \{\chi_t > 0\}\subseteq\mathcal{E}_t \text{ for all } t\},\\
 \mathcal{T}^\mathcal{E} &:= \{\tau\in\mathcal{T} :  \{\tau =t \}\subseteq\mathcal{E}_t \text{ for all } t\}.
\end{align*}
The following result specifies the relationship between $\mathcal{E}$, $\mathcal{T}^\mathcal{E}$ and $\mathcal{X}^\mathcal{E}$.

\begin{proposition} \label{prop:E-equivalent-T^E}
 For all $t$, \[\mathcal{E}_t=\bigcup_{\tau\in\mathcal{T}^\mathcal{E}}\{\tau=t\}=\bigcup_{\chi\in\mathcal{X}^\mathcal{E}}\{\chi_t>0\}.\]
\end{proposition}

\begin{proof}
 For the first equality, it is clear from the definition of $\mathcal{T}^\mathcal{E}$ that
\[
 \bigcup_{\tau\in\mathcal{T}^\mathcal{E}}\{\tau=t\} \subseteq \mathcal{E}_t \text{ for all } t.
\]
We now show for any $t^\prime=0,\ldots,T$ that there exists a stopping time $\tau^\prime\in\mathcal{T}^\mathcal{E}$ such that $\{\tau^\prime=t^\prime\}=\mathcal{E}_{t^\prime}$.
Define
\[
 \mathcal{E}^\prime_t :=
\begin{cases}
 \mathcal{E}_t\setminus\mathcal{E}^\ast_{t+1} &\text{if } t<t^\prime,\\
 \mathcal{E}_{t^\prime}  &\text{if } t=t^\prime,\\
 \mathcal{E}_t\setminus\bigcup_{s=0}^{t-1}\mathcal{E}^\prime_s &\text{if } t>t^\prime,
\end{cases}
\]
so that $\mathcal{E}^\prime_0,\ldots,\mathcal{E}^\prime_T$ is a sequence of mutually disjoint sets in $\Omega$ with $\mathcal{E}^\prime_t\subseteq\mathcal{E}_t$ and $\mathcal{E}^\prime_t\in\mathcal{F}_t$ for all $t$. Moreover it is a partition of $\Omega$ since
\begin{align*}
 \bigcup_{t=0}^T\mathcal{E}^\prime_t
&= \left\{\bigcup_{t=0}^{t^\prime-1}\mathcal{E}^\prime_t\right\}\cup\mathcal{E}^\prime_{t^\prime} \cup \left\{\mathcal{E}^\ast_{t^\prime+1}\setminus\bigcup_{t=0}^{t^\prime}\mathcal{E}^\prime_t\right\} \\
&= \left\{\bigcup_{t=0}^{t^\prime-1}\left[\mathcal{E}_t\setminus\mathcal{E}^\ast_{t+1}\right]\right\}\cup\mathcal{E}^\ast_{t^\prime}
= \Omega.
\end{align*}
The random variable
\[
 \tau^\prime := \sum_{t=0}^T t\mathbf{1}_{\mathcal{E}^\prime_t}
\]
is therefore a stopping time in $\mathcal{T}^\mathcal{E}$ with $\{\tau^\prime = t^\prime\} = \mathcal{E}^\prime_{t^\prime} = \mathcal{E}_{t^\prime}$ as required.

The second equality holds because
\[
 \mathcal{E}_t =
\bigcup_{\tau\in\mathcal{T}^\mathcal{E}}\{\tau=t\} \subseteq
\bigcup_{\chi\in\mathcal{X}^\mathcal{E}}\{\chi_t>0\} \subseteq
\mathcal{E}_t
\]
for all $t$.
\end{proof}

\section{Main results and discussion}
\label{sec:main-results}

An \emph{option} consists of an adapted $\mathbb{R}^d$-valued payoff process $\xi=(\xi_t)$ and an exercise policy $\mathcal{E}$. The seller delivers the portfolio $\xi_\tau\in\mathcal{L}_\tau$ to the buyer at a stopping time $\tau\in\mathcal{T}^\mathcal{E}$ chosen by the buyer among the stopping times consistent with $\mathcal{E}$.

\subsection{Pricing and hedging for the seller}

Consider the hedging and pricing problem for the seller of the option $(\xi,\mathcal{E})$. A self-financing trading strategy $y\in\Phi$ is said to \emph{superhedge $(\xi,\mathcal{E})$ for the seller} if
\begin{equation} \label{eq:def:super-hedge}
 y_\tau-\xi_\tau\in\mathcal{K}_\tau \text{ for all } \tau\in\mathcal{T}^\mathcal{E}.
\end{equation}

\begin{definition}[Ask price]
 The \emph{ask price} or \emph{seller's price} or \emph{upper hedging price} of~$(\xi,\mathcal{E})$ at time $0$ in terms of any currency $i=1,\ldots,d$ is defined as
\begin{equation} \label{eq:pi-a}
 p^a_i(\xi,\mathcal{E}) := \inf\{x\in\mathbb{R}:y\in\Phi\text{ with }y_0=xe^i\text{ superhedges }(\xi,\mathcal{E})\text{ for the seller}\}.
\end{equation}
\end{definition}

The interpretation of the ask price is that an endowment of at least~\(p^a_i(\xi,\mathcal{E})\) units of asset~\(i\) at time~\(0\) would enable an investor to settle the option without risk. A superhedging strategy $y$ for the seller is called \emph{optimal} if $y_0=p^a_i(\xi,\mathcal{E})e^i$.

Our main aims are to compute the option price~$p^a_i(\xi,\mathcal{E})$ algorithmically and to find a probabilistic dual representation for it, to construct the set of initial endowments that allow the seller to superhedge, and to construct an optimal superhedging strategy~$y\in\Phi$ for the seller. To this end, consider the following construction.

\begin{construction} \label{const:seller}
 For all $t$, let
\begin{equation} \label{eq:construction:Ut}
                                                         \mathcal{U}^a_t :=
\begin{cases}
 \xi_t + \mathcal{K}_t & \text{on } \mathcal{E}_t,\\
\mathbb{R}^d & \text{on }\Omega\setminus\mathcal{E}_t.
\end{cases}\end{equation}
Define
\begin{align*}
 \mathcal{V}^a_T := \mathcal{W}^a_T &:= \mathcal{L}_T,\\
 \mathcal{Z}^a_T &:= \mathcal{U}^a_T.
\end{align*}
For $t<T$, let
\begin{align}
 \mathcal{W}^a_t &:= \mathcal{Z}^a_{t+1} \cap \mathcal{L}_t,\nonumber\\
  \mathcal{V}^a_t &:= \mathcal{W}^a_t + \mathcal{K}_t,\label{eq:Wa:def}\\
 \mathcal{Z}^a_t &:=\mathcal{U}^a_t\cap\mathcal{V}^a_t.\label{eq:Za:def}
\end{align}
\end{construction}

For each $t$ the set $\mathcal{U}^a_t$ is the collection of portfolios in $\mathcal{L}_t$ that allows the seller to settle the option at time $t$. We shall demonstrate in {Proposition~\ref{prop:seller:equivalence-construction} that} for each $t<T$ the sets $\mathcal{V}^a_t$, $\mathcal{W}^a_t$ and $\mathcal{Z}^a_t$ have natural interpretations as collections of portfolios that are of importance to the seller of the option. The set $\mathcal{W}^a_t$ is the collection of portfolios at time $t$ that allow the seller to settle the option in the future (at time $t+1$ or later). The set $\mathcal{V}^a_t$ consists of those portfolios that may be rebalanced at time $t$ into a portfolio in $\mathcal{W}^a_t$, and $\mathcal{Z}^a_t$ consists of all portfolios that allow the seller to {remain solvent after settling} the option at time $t$ or any time in the future.

\begin{remark} \label{rem:const:seller:when-is-R^d}
 On $\mathcal{E}_t$, where exercise is allowed, the set $\mathcal{U}^a_t$ is a translation of~$\mathcal{K}_t$, so it is non-empty and polyhedral. It is then straightforward to show by backward induction that the following holds for all $t$:
\begin{itemize}
 \item  $\mathcal{V}^a_t$, $\mathcal{W}^a_t$, $\mathcal{Z}^a_t$ are all non-empty.
 \item $\mathcal{V}^a_t=\mathcal{W}^a_t=\mathcal{Z}^a_t=\mathbb{R}^d$ on $\Omega\setminus\mathcal{E}^\ast_t$.
 \item $\mathcal{Z}^a_t=\mathcal{V}^a_t$ on $\Omega\setminus\mathcal{E}_t$ and $\mathcal{Z}^a_t=\mathcal{U}^a_t$ on $\Omega\setminus\mathcal{E}^\ast_{t+1}$.
 \item $\mathcal{V}^a_t$ and $\mathcal{W}^a_t$ are polyhedral on $\mathcal{E}^\ast_{t+1}$ and $\mathcal{Z}^a_t$ is polyhedral on $\mathcal{E}^\ast_t$.
\end{itemize}
Note in particular that the non-empty set $\mathcal{Z}^a_0$ is polyhedral since $\mathcal{E}^\ast_0=\Omega$.
\end{remark}

The main pricing and hedging result for the seller reads as follows.

\begin{theorem} \label{th:seller}
The set $\mathcal{Z}^a_0$ is the collection of initial endowments allowing the seller to superhedge $(\xi,\mathcal{E})$, and
\begin{align*}
 p^a_i(\xi,\mathcal{E})
&= \max_{\chi\in\mathcal{X}^\mathcal{E}}\max_{(\mathbb{Q},S)\in\bar{\mathcal{P}}^i(\chi)} \mathbb{E}_\mathbb{Q}((\xi\cdot S)_\chi)\\
&= \max_{\chi\in\mathcal{X}^\mathcal{E}}\sup_{(\mathbb{Q},S)\in\mathcal{P}^i(\chi)} \mathbb{E}_\mathbb{Q}((\xi\cdot S)_\chi)\\
&= \min\{x\in\mathbb{R}:xe^i\in\mathcal{Z}^a_0\}\\
&= -\min\{Z^a_0(s):s\in\sigma_i(\mathbb{R}^d)\},
\end{align*}
where $Z^a_0$ is the support function of $-\mathcal{Z}_0^a$. An optimal superhedging strategy $y\in\Phi$ for the seller can be constructed algorithmically, and so can a randomised stopping time $\hat{\chi}\in\mathcal{X}^\mathcal{E}$ and $\hat{\chi}$-approximate martingale pair $(\hat{\mathbb{Q}},\hat{S})\in\bar{\mathcal{P}}^i(\hat{\chi})$ such that
\begin{equation} \label{eq:seller:1}
 \mathbb{E}_{\hat{\mathbb{Q}}}((\xi\cdot \hat{S})_{\hat{\chi}}) = p^a_i(\xi,\mathcal{E}).
\end{equation}
\end{theorem}

Any stopping time $\hat{\chi}$ and $\hat{\chi}$-approximate martingale pair $(\hat{\mathbb{Q}},\hat{S})\in\bar{\mathcal{P}}^i(\hat{\chi})$ satisfying \eqref{eq:seller:1} are called \emph{optimal for the seller} of $(\xi,\mathcal{E})$. Note that the optimal superhedging strategy, {stopping time and approximate martingale pair are not unique in general}.

The proof of Theorem~\ref{th:seller} appears in {Section~\ref{sec:seller}, together} {with details of the construction of the optimal stopping time and approximate martingale pair for the seller. An optimal superhedging strategy can be found using the following construction with initial value $y_0=p^a_i(\xi,\mathcal{E})e^i$.}

\begin{construction} \label{const:seller:hedge}
  Take $y_0\in\mathcal{Z}^a_0$ as given. For all $t<T$ choose any
  \begin{equation}\label{eq:const:seller:hedge}
  y_{t+1}\in(y_t-\mathcal{K}_t)\cap\mathcal{W}^a_t.
  \end{equation}
 \end{construction}

 The correctness of Construction \ref{const:seller:hedge} will be demonstrated in Proposition \ref{prop:seller:equivalence-construction} below. Note that if the set $(y_t-\mathcal{K}_t)\cap\mathcal{W}^a_t$ in \eqref{eq:const:seller:hedge} is not a singleton, then there is some freedom as to the choice of $y_{t+1}$.

\subsection{Pricing and hedging for the buyer}

Consider now the pricing and hedging problem for the buyer of the option $(\xi,\mathcal{E})$. A pair $(y,\tau)$ consisting of a self-financing trading strategy $y\in\Phi$ and a stopping time $\tau\in\mathcal{T}^\mathcal{E}$ \emph{superhedges $(\xi,\mathcal{E})$ for the buyer} if
\begin{equation} \label{eq:def:super-hedge:buyer}
 y_\tau + \xi_\tau \in \mathcal{K}_\tau.
\end{equation}

\begin{definition}[Bid price]
The \emph{bid price} or \emph{buyer's price} or \emph{lower hedging price} of $(\xi,\mathcal{E})$ at time $0$ in terms of currency {$i=1,\ldots,d$} is defined as
\begin{equation} \label{eq:pi-b}
 p^b_i(\xi,\mathcal{E})
:= \sup\{-x\in\mathbb{R}:(y,\tau)\text{ with }y_0=xe^i\text{ superhedges }(\xi,\mathcal{E})\text{ for the buyer}\}.
\end{equation}
\end{definition}

The interpretation of the bid price is that $p^b_i(\xi,\mathcal{E})$ is the largest amount in currency~$i$ that can be raised at time $0$ by the owner of {the option} $(\xi,\mathcal{E})$ by setting up a self-financing trading strategy with the property that it leaves him in a solvent position after exercising $(\xi,\mathcal{E})$.  A superhedging strategy $(y,\tau)$ for the buyer is called \emph{optimal} if $y_0=-p^b_i(\xi,\mathcal{E})e^i$.

Just as in the seller's case, the aims are to {algorithmically} compute $p^b_i(\xi,\mathcal{E})$, to establish a probabilistic representation for it, to find the set of initial endowments allowing superhedging for the buyer, and to construct an optimal superhedging strategy for the buyer. The key to this is the following construction.

\begin{construction}\label{const:buyer}
 For all $t$, let
\[
                                                         \mathcal{U}^b_t :=
\begin{cases}
 -\xi_t + \mathcal{K}_t & \text{on } \mathcal{E}_t,\\
\emptyset & \text{on }\Omega\setminus\mathcal{E}_t.
\end{cases}\]
Define
\begin{align*}
 \mathcal{V}^b_T := \mathcal{W}^b_T &:= \emptyset,\\
 \mathcal{Z}^b_T &:= \mathcal{U}^b_T.
\end{align*}
For $t<T$, let
\begin{align}
 \mathcal{W}^b_t &:= \mathcal{Z}^b_{t+1} \cap \mathcal{L}_t^d,\nonumber\\
  \mathcal{V}^b_t &:= \mathcal{W}^b_t + \mathcal{K}_t,\nonumber\\
 \mathcal{Z}^b_t &:=\mathcal{U}^b_t\cup\mathcal{V}^b_t. \label{eq:Zb:def}
\end{align}
\end{construction}

For each $t$ the set $\mathcal{U}^b_t$ is the collection of portfolios in $\mathcal{L}_t$ that allows the buyer to be in a solvent position after exercising the option at time $t$. For $t<T$ the set $\mathcal{W}^b_t$ is the collection of portfolios at time $t$ that allow the buyer to superhedge the option in the future (at time $t+1$ or later), and $\mathcal{V}^b_t$ consists of those portfolios that may be rebalanced at time $t$ into a portfolio in $\mathcal{W}^b_t$. The set  $\mathcal{Z}^b_t$ consists of all portfolios that allow the buyer to {remain solvent after exercising} the option at time $t$ or any time in the future.

\begin{remark}
Construction \ref{const:buyer} differs from Construction \ref{const:seller} in two respects. Firstly, the payoff is treated differently because it is delivered by the seller and received by the buyer. Secondly, there is a union {of sets} in \eqref{eq:Zb:def} where there is an intersection in \eqref{eq:Za:def}. This encapsulates the opposing positions of the seller and the buyer: any portfolio held by the seller at time $t$ must enable him to settle the option at time $t$ or later, whereas any portfolio held by the buyer needs to enable him to achieve solvency by exercising the option, either at time~$t$ or at some point in the future. The union in \eqref{eq:Zb:def} also illustrates the fact that the pricing problem for the buyer is not convex.
\end{remark}

\begin{remark} \label{rem:const:buyer:when-is-empty} On $\mathcal{E}_t$ the set $\mathcal{U}^b_t$ is polyhedral and non-empty. It is possible to show the following by backward induction on $t$:
 \begin{itemize}
  \item $\mathcal{V}^b_t=\mathcal{W}^b_t=\mathcal{Z}^b_t=\emptyset$ on $\Omega\setminus\mathcal{E}^\ast_t$.

  \item $\mathcal{Z}^b_t=\mathcal{V}^b_t$ on $\Omega\setminus\mathcal{E}_t$ and $\mathcal{Z}^b_t=\mathcal{U}^b_t$ on $\Omega\setminus\mathcal{E}^\ast_{t+1}$.

  \item $\mathcal{V}^b_t$ and $\mathcal{W}^b_t$ on $\mathcal{E}^\ast_{t+1}$, and $\mathcal{Z}^b_t$ on $\mathcal{E}^\ast_t$ can be written as a finite union of non-empty closed polyhedral sets (but $\mathcal{V}^b_t$, $\mathcal{W}^b_t$ and $\mathcal{Z}^b_t$ are not convex in general).
 \end{itemize}
Note in particular that since $\mathcal{E}^\ast_0=\Omega$, the last item applies to $\mathcal{Z}^b_0$, so it is non-empty and closed.
\end{remark}

Here is the main pricing and hedging theorem for the buyer.

\begin{theorem} \label{th:buyer}
The set $\mathcal{Z}^b_0$ is the collection of initial endowments allowing superhedging of $(\xi,\mathcal{E})$ by the buyer, and
\begin{align*}
 p^b_i(\xi,\mathcal{E})
&= \max_{\tau\in\mathcal{T}^\mathcal{E}}\min_{(\mathbb{Q},S)\in\bar{\mathcal{P}}^i(\tau)} \mathbb{E}_\mathbb{Q}((\xi\cdot S)_\tau)\\
&= \max_{\tau\in\mathcal{T}^\mathcal{E}}\inf_{(\mathbb{Q},S)\in\mathcal{P}^i(\tau)} \mathbb{E}_\mathbb{Q}((\xi\cdot S)_\tau)\\
&= -\min\{x\in\mathbb{R}|xe^i\in\mathcal{Z}^b_0\}.
\end{align*}
An optimal superhedging strategy $(\check{y},\check{\tau})\in\Phi\times\mathcal{T}^\mathcal{E}$ with
\[
 \check{\tau} = \min\{t:\check{y}_t\in\mathcal{U}^b_t\}
\]
can be constructed algorithmically, and so can a $\check{\tau}$-approximate martingale pair $(\check{\mathbb{Q}},\check{S})\in\bar{\mathcal{P}}^i(\check{\tau})$ such that
\begin{equation} \label{eq:th:buyer:tau-approx-pair}
 \mathbb{E}_{\check{\mathbb{Q}}}((\xi\cdot \check{S})_{\check{\tau}}) = p^b_i(\xi,\mathcal{E}).
\end{equation}
\end{theorem}

A $\check{\tau}$-approximate martingale pair $(\check{\mathbb{Q}},\check{S})$ is called \emph{optimal for the buyer} if it satisfies~\eqref{eq:th:buyer:tau-approx-pair}. The proof of this theorem appears in Section \ref{sec:buyer}, together with full details of the construction of an optimal $\check{\tau}$-approximate martingale pair. An optimal superhedging strategy $(y,\tau)$ can be obtained by means of the following construction with initial choice $y_0=-p^b_i(\xi,\mathcal{E})e^i$ and with $\tau:=\tau_T$.

\begin{construction}\label{const:buyer:hedge}
Take $y_0\in\mathcal{Z}^b_0$ as given, and define
\[
 \tau_0 :=
\begin{cases}
 0 & \text{if } y_0\in\mathcal{U}^b_0,\\
 1 &\text{if } y_0\in\mathcal{Z}^b_0\setminus\mathcal{U}^b_0.
\end{cases}.
\]
For $t<T$, choose any
\begin{equation} \label{eq:const:buyer:hedge:1}
 y_{t+1}\in
 \begin{cases}
  \{y_t\}&\text{ on }\{\tau_t\le t\},\\
  (y_t-\mathcal{K}_t)\cap\mathcal{W}^b_t &\text{ on }\{\tau_t=t+1\},
 \end{cases}
\end{equation}
and define
\begin{equation} \label{eq:const:buyer:hedge:2}
 \tau_{t+1} :=
\begin{cases}
 \tau_t & \text{on } \{\tau_t\le t\},\\
 t+1 & \text{on } \{\tau_t=t+1\} \cap \{y_{t+1}\in\mathcal{U}^b_{t+1}\},\\
 t+2 & \text{on } \{\tau_t=t+1\} \cap \{y_{t+1}\in\mathcal{Z}^b_{t+1}\setminus\mathcal{U}^b_{t+1}\}.
\end{cases}
\end{equation}
\end{construction}

The correctness of Construction \ref{const:buyer:hedge} will be established in detail in Proposition \ref{prop:buyer:equivalence-construction}. Note that, similar to the seller's case, there may be some choice in the construction of this strategy due to the fact that $(y_t-\mathcal{K}_t)\cap\mathcal{W}^b_t$ in \eqref{eq:const:buyer:hedge:1} may contain more than one element.

\subsection{Special cases}

\subsubsection{European options}
\label{sec:European}

Consider a European-style option that offers the payoff $\zeta\in\mathcal{L}_\tau$ at some given stopping time $\tau\in\mathcal{T}$ (in particular, we can have $\tau=T$ for an ordinary European option with expiry time $T$). Here $\mathcal{L}_\tau$ is the set of $\mathbb{R}^d$-valued $\mathcal{F}_\tau$-measurable random variables. In our framework the payoff of such an option is the adapted process $\xi=(\xi_t)$ with
\[
 \xi_t = \zeta\mathbf{1}_{\{\tau=t\}} \text{ for all } t,
\]
and its exercise policy $\mathcal{E}=(\mathcal{E}_t)$ is given by
\[
\mathcal{E}_t:=\{\tau=t\} \text{ for all }t.
\]
It follows that $\mathcal{T}^\mathcal{E}=\{\tau\}$ and ${\mathcal{X}}^\mathcal{E}=\{\chi^\tau\}$. For clarity we denote this European option by $(\zeta,\tau)$ instead of $(\xi,\mathcal{E})$.

Observe that a trading strategy $y\in\Phi$ superhedges the option $(\zeta,\tau)$ for the seller if and only if $(y,\tau)$ superhedges $(-\zeta,\tau)$ for the buyer. It also follows directly from \eqref{eq:pi-a} and \eqref{eq:pi-b} that
\[
 p^b_i(\zeta,\tau) = -p^a_i(-\zeta,\tau).
\]
Thus the pricing and hedging problems for the buyer and seller of a European-style option are symmetrical. In particular, this means that the pricing problem for the buyer is convex, and the hedging problem for the seller does not involve any randomised stopping times.

Constructions \ref{const:seller} and \ref{const:buyer} can be simplified considerably due to the simple structure of the exercise policy. Noting that at each time step $t$ we have  $\mathcal{U}^a_t=\zeta+\mathcal{K}_t$ on $\{t=\tau\}$ and $\mathcal{U}^a_t=\mathbb{R}^d$ on $\{t\neq\tau\}$, Construction \ref{const:seller} can now be rewritten as follows for each $t$:
\begin{align}
 \mathcal{Z}^a_t&=\mathbb{R}^d &\text{on } \{t>\tau\}, \nonumber\\
 \mathcal{Z}^a_t&=\zeta+\mathcal{K}_t &\text{on } \{t=\tau\}, \label{eq:E:const:1}\\
 \mathcal{Z}^a_t&=(\mathcal{Z}^a_{t+1}\cap\mathcal{L}_t)+\mathcal{K}_t &\text{on } \{t<\tau\},\label{eq:E:const:2}
\end{align}
where the auxiliary sets $\mathcal{W}^a_t,\mathcal{V}^a_t$ are omitted, for simplicity.
Theorem \ref{th:seller} gives the ask price of $(\zeta,\tau)$ as
\begin{equation}
p^a_i(\zeta,\tau)
= \sup_{(\mathbb{Q},S)\in\mathcal{P}^i(\chi^\tau)} \mathbb{E}_\mathbb{Q}((\xi\cdot S)_{\chi^\tau})
= \sup_{(\mathbb{Q},S)\in\mathcal{P}^i(\tau)} \mathbb{E}_\mathbb{Q}(\zeta\cdot S_\tau)\label{eq:ask:E}
\end{equation}
since $\mathcal{P}(\tau)=\mathcal{P}(\chi^\tau)$ and
$ (\xi\cdot S)_{\chi^\tau} = (\xi\cdot S)_\tau = \zeta\cdot S_\tau$. A similar simplification is possible for the buyer; note that $\mathcal{V}^b_t$, $\mathcal{W}^b_t$ and $\mathcal{Z}^b_t$ are convex for all $t$. The bid price  of $(\zeta,\tau)$ is
\begin{align*}
  p^b_i(\zeta,\tau)& = -p^a_i(-\zeta,\tau) = \inf_{(\mathbb{Q},S)\in\mathcal{P}^i(\tau)} \mathbb{E}_\mathbb{Q}(\zeta\cdot S_{\tau}),
\end{align*}
which is consistent with Theorem \ref{th:buyer}.

Consider the special case $\tau=T$, which corresponds to a classical European option. The simplified construction \eqref{eq:E:const:1}--\eqref{eq:E:const:2} {leads to the same set $\mathcal{Z}^a_t$ of superhedging portfolios} as in {\cite[Theorem~2]{Loehne_Rudloff2011}}. The representations for the bid and ask prices can be simplified further by noting that
\[
S_{t+1}^{\chi^T\ast} = S_T \text{ for all }t<T.
\]
For any $(\mathbb{Q},S)\in\bar{\mathcal{P}}^i(\chi^T)$, the adapted process $\check{S}=(\check{S}_t)$ defined by
\[
    \check{S}_t:= \mathbb{E}_{\mathbb{Q}}(S_T|\mathcal{F}_t) \text{ for all }t
\]
is a $\mathbb{Q}$-martingale such that $(\mathbb{Q},\check{S})\in\bar{\mathcal{P}}^i$, and
\[
 \mathbb{E}_\mathbb{Q}(\zeta\cdot S_T) = \mathbb{E}_\mathbb{Q}(\zeta\cdot \check{S}_T).
\]
Thus the supremum in \eqref{eq:ask:E} need only be taken over $\mathcal{P}^i$, and it follows that
\begin{align*}
 p^a_i(\zeta,T)& = \sup_{(\mathbb{Q},S)\in\mathcal{P}^i} \mathbb{E}_\mathbb{Q}(\zeta\cdot S_T)= \max_{(\mathbb{Q},S)\in\bar{\mathcal{P}}^i} \mathbb{E}_\mathbb{Q}(\zeta\cdot S_T).
\end{align*}
This result extends \cite{bensaid_lesne_pages_scheinkman1992,edirisinghe_naik_uppal1993,Roux_Tokarz_Zastawniak2008} in two-asset models. Its conclusions are technically closest to the non-constructive results for currency models in {\cite{delbaen_kabanov_valkeila2002,kabanov_stricker2001b}}.

\subsubsection{Bermudan options}

The exercise policy $\mathcal{E}$ for a Bermudan option with payoff process $\xi$ that can be exercised at given times $t_1<\cdots<t_n$ is defined in \eqref{eq:bermudan-exercise-policy}. The collections of ordinary and randomised stopping times consistent with this exercise policy are
\begin{align*}
 \mathcal{T}^\mathcal{E} &= \{\tau\in\mathcal{T}:\tau\in\{t_1,\ldots,t_n\}\}, \\
 \mathcal{X}^\mathcal{E} &= \{\chi\in\mathcal{X}:\chi_t=0\text{ for all } t\notin\{t_1,\ldots,t_n\}\}.
\end{align*}
Note that $\mathcal{Z}^a_t$ and $\mathcal{Z}^b_t$ are closed non-empty strict subsets of $\mathcal{L}_t$ whenever $t\le t_n$. Theorems \ref{th:seller} and \ref{th:buyer} can then be used to compute $p^a_i(\xi,\mathcal{E})$ and $p^b_i(\xi,\mathcal{E})$. Moreover optimal superhedging strategies $y^a\in\Phi$ for the seller and $(y^b,\tau)\in\Phi\times\mathcal{T}^\mathcal{E}$ for the buyer can be constructed algorithmically.

\subsubsection{American options}\label{sec:American}

Consider an American option with expiration date $T$ that offers the payoff~$\xi_\tau$ at a stopping time $\tau\in\mathcal{T}$ chosen by the buyer.
The exercise policy $\mathcal{E}=(\mathcal{E}_t)$ satisfies $\mathcal{E}_t=\Omega$ for all $t\le T$, and the sets of stopping times consistent with the exercise policy are
\[
 \mathcal{T}^\mathcal{E}=\mathcal{T},\quad\quad
 \mathcal{X}^\mathcal{E}=\mathcal{X}.
\]
Denote this American-style option by~$\xi$ instead of $(\xi,\mathcal{E})$. Theorems~\ref{th:seller} and~\ref{th:buyer} give the ask and bid prices as
\begin{align*}
 p^a_i(\xi) &=\max_{\chi\in\mathcal{X}}\sup_{(\mathbb{Q},S)\in\mathcal{P}^i(\chi)} \mathbb{E}_\mathbb{Q}((\xi\cdot S)_\chi)=\max_{\chi\in\mathcal{X}}\max_{(\mathbb{Q},S)\in\bar{\mathcal{P}}^i(\chi)} \mathbb{E}_\mathbb{Q}((\xi\cdot S)_\chi),\\
 p^b_i(\xi) &=\max_{\tau\in\mathcal{T}}\inf_{(\mathbb{Q},S)\in\mathcal{P}^i(\tau)} \mathbb{E}_\mathbb{Q}((\xi\cdot S)_\tau)=\max_{\tau\in\mathcal{T}}\min_{(\mathbb{Q},S)\in\bar{\mathcal{P}}^i(\tau)} \mathbb{E}_\mathbb{Q}((\xi\cdot S)_\tau).
\end{align*}
This directly extends \cite{chalasani_jha2001,Roux_Zastawniak2009,tokarz_zastawniak2006} for two-asset models. In the context of currency models, this is consistent with the results in \cite{bouchard_temam2005} for the seller.

\begin{remark}
  In this work it is assumed that trading strategies are rebalanced at each time instant $t$ only after it becomes known that the option is not to be exercised at that time instant. In their work on pricing American options for the seller, Bouchard and Temam \cite{bouchard_temam2005} follow a different convention by assuming that the portfolios in a hedging strategy must be rebalanced before exercise decisions become known. The method in this paper also applies to their case, provided that the order of the operations in~\eqref{eq:Wa:def} and~\eqref{eq:Za:def} is interchanged, i.e.\ replace these equations by \[\mathcal{Z}^a_{t}:=\mathcal{W}^a_{t}\cap\mathcal{U}^a_{t}+\mathcal{K}^a_{t} \text{ for }t<T.\]
	Ask prices obtained in this way are in general higher than the ask prices presented above. This is because a superhedging strategy for the seller in this setting will also superhedge under our definition, but the converse is not always true. Because of this, superhedging as we have defined above is easier to achieve, and it is therefore more natural for traders to follow than the approach of Bouchard and Temam.
\end{remark}

Examples \ref{Exl:2} and \ref{Exl:3} below demonstrate the computation of the bid and ask prices of American options in toy models, and Examples \ref{ex:1} and \ref{ex:2} demonstrate the same in models with a more realistic flavour.

\section{Pricing and hedging for the seller}
\label{sec:seller}

This section is devoted to the proof of Theorem \ref{th:seller}. Recall that a trading strategy $y\in\Phi$ superhedges the option $(\xi,\mathcal{E})$ for the seller if \eqref{eq:def:super-hedge} holds. In view of Proposition~\ref{prop:E-equivalent-T^E}, this is equivalent to
\[
 y_t - \xi_t \in\mathcal{K}_t \text{ on } \mathcal{E}_t \text{ for all } t
\]
or
\begin{equation} \label{eq:def:super-hedge:equiv:2}
 y_t \in\mathcal{U}^a_t \text{ for all } t.
\end{equation}
We now have the following result.

\begin{proposition} \label{prop:seller-superhedges-exp}
The ask price $p^a_i(\xi,\mathcal{E})$ defined in \eqref{eq:pi-a} is finite and
\begin{align}
 p^a_i(\xi,\mathcal{E})
&\ge \sup_{\chi\in\mathcal{X}^\mathcal{E}}\max_{(\mathbb{Q},S)\in\bar{\mathcal{P}}^i(\chi)} \mathbb{E}_\mathbb{Q}((\xi\cdot S)_\chi)\label{eq:prop:seller-superhedges-exp:eq}\\
&= \sup_{\chi\in\mathcal{X}^\mathcal{E}}\sup_{(\mathbb{Q},S)\in\mathcal{P}^i(\chi)} \mathbb{E}_\mathbb{Q}((\xi\cdot S)_\chi)\label{eq:prop:seller-superhedges-exp:ineq}.
\end{align}
\end{proposition}

\begin{proof}
We show by backward induction below that if $\chi\in\mathcal{X}^\mathcal{E}$, $(\mathbb{Q},S)\in\bar{\mathcal{P}}^i(\chi)$ and $y\in\Phi$ with  $y_0=xe^i$ superhedges $(\xi,\mathcal{E})$ for the seller, then
\begin{equation} \label{eq:prop:seller-superhedges-exp:induction}
 y_t\cdot\mathbb{E}_\mathbb{Q}(S_t^{\chi\ast}|\mathcal{F}_t) \ge \mathbb{E}_\mathbb{Q}((\xi\cdot S)_t^{\chi\ast}|\mathcal{F}_t)
\end{equation}
for all $t$. The property $S^i\equiv1$ then gives
\[
 x = xe^i\cdot\mathbb{E}_\mathbb{Q}(S_\chi) = xe^i\cdot\mathbb{E}_\mathbb{Q}(S_0^{\chi\ast}|\mathcal{F}_0) \ge \mathbb{E}_\mathbb{Q}((\xi\cdot S)_0^{\chi\ast}|\mathcal{F}_0) = \mathbb{E}_\mathbb{Q}((\xi\cdot S)_\chi),
\]
and the inequality \eqref{eq:prop:seller-superhedges-exp:eq} is immediate. The equality \eqref{eq:prop:seller-superhedges-exp:ineq} follows directly from Lemma~\ref{lem:closure}. The property $-\infty < p^a_i(\xi,\mathcal{E}) < \infty$ holds true since $\mathcal{P}^i(\chi)\neq\emptyset$ and $(\xi,\mathcal{E})$ has a trivial superhedging strategy for the seller, given by $(y^1_t,\ldots,y^d_t)$ where
\begin{equation} \label{eq:seller:trivial-superhedging-strategy}
 y^j_t := \max\{\xi^{j\omega}_s:s=0,\ldots,T,\omega\in\mathcal{E}_s\}
\end{equation}
for all $j$ and $t$.

Observe that for any $t$ we have $\chi_t=0$ on $\Omega\setminus\mathcal{E}_t$, and on $\mathcal{E}_t$ we have $y_t-\xi_t\in\mathcal{K}_t$, so that $y_t\cdot S_t \ge \xi_t\cdot S_t$ since $S_t\in\mathcal{K}^\ast_t$ (see Definition \ref{def:chi-eqv-pricing-measure-and-pricing-process}). This means that
\[
 \chi_t y_t\cdot S_t \ge \chi_t \xi_t\cdot S_t \text{ for all } t.
\]
To prove \eqref{eq:prop:seller-superhedges-exp:induction} by backward induction, first note that at time $T$,
\[
 y_T\cdot\mathbb{E}_\mathbb{Q}(S_T^{\chi\ast}|\mathcal{F}_T) = \chi_Ty_T\cdot S_T \ge \chi_T\xi_T\cdot S_T = \mathbb{E}_\mathbb{Q}((\xi\cdot S)_T^{\chi\ast}|\mathcal{F}_T).
\]
Suppose for some $t<T$ that
\[
 y_{t+1}\cdot\mathbb{E}_\mathbb{Q}(S_{t+1}^{\chi\ast}|\mathcal{F}_{t+1}) \ge \mathbb{E}_\mathbb{Q}((\xi\cdot S)_{t+1}^{\chi\ast}|\mathcal{F}_{t+1}).
\]
The self-financing condition $y_t - y_{t+1}\in\mathcal{K}_t$ together with $\mathbb{E}_\mathbb{Q}(S_{t+1}^{\chi\ast}|\mathcal{F}_t)\in\mathcal{K}^\ast_t$ (see Definition \ref{def:chi-eqv-pricing-measure-and-pricing-process}) gives
\[
 y_t\cdot \mathbb{E}_\mathbb{Q}(S_{t+1}^{\chi\ast}|\mathcal{F}_t) \ge y_{t+1}\cdot \mathbb{E}_\mathbb{Q}(S_{t+1}^{\chi\ast}|\mathcal{F}_t).
\]
Combining this with the inductive assumption, we obtain
\begin{align*}
 y_t\cdot \mathbb{E}_\mathbb{Q}(S_t^{\chi\ast}|\mathcal{F}_t)
&=  \chi_ty_t\cdot S_t + y_t\cdot \mathbb{E}_\mathbb{Q}(S_{t+1}^{\chi\ast}|\mathcal{F}_t) \\
&\ge  \chi_t\xi_t\cdot S_t + y_{t+1}\cdot \mathbb{E}_\mathbb{Q}(S_{t+1}^{\chi\ast}|\mathcal{F}_t) \\
&=  \mathbb{E}_\mathbb{Q}(\chi_t\xi_t\cdot S_t + y_{t+1}\cdot \mathbb{E}_\mathbb{Q}(S_{t+1}^{\chi\ast}|\mathcal{F}_{t+1})|\mathcal{F}_t) \\
&\ge  \mathbb{E}_\mathbb{Q}(\chi_t\xi_t\cdot S_t + \mathbb{E}_\mathbb{Q}((\xi\cdot S)_{t+1}^{\chi\ast}|\mathcal{F}_{t+1})|\mathcal{F}_t) \\
&=  \mathbb{E}_\mathbb{Q}(\chi_t\xi_t\cdot S_t + (\xi\cdot S)_{t+1}^{\chi\ast}|\mathcal{F}_t) \\
&=  \mathbb{E}_\mathbb{Q}((\xi\cdot S)_t^{\chi\ast}|\mathcal{F}_t).
\end{align*}
This concludes the inductive step.
\end{proof}

The next result shows that $\mathcal{Z}^a_0$ is the set of initial endowments of self-financing trading strategies that allow the seller to superhedge $(\xi,\mathcal{E})$. It also links Construction \ref{const:seller} with the problem of computing the ask price in \eqref{eq:pi-a}.

\begin{proposition} \label{prop:seller:equivalence-construction}
We have
\[
 \mathcal{Z}_0^a = \{y_0\in\mathbb{R}^d:y=(y_t)\in\Phi \text{ superhedges } (\xi,\mathcal{E}) \text{ for the seller}\}
\]
and
\begin{equation} \label{eq:prop:seller:equivalence-construction:1}
 p^a_i(\xi,\mathcal{E}) = \min\{x\in\mathbb{R}|xe^i\in\mathcal{Z}^a_0\}.
\end{equation}
Moreover, Construction \ref{const:seller:hedge} yields a superhedging strategy $y\in\Phi$ for the seller for any $y_0\in\mathcal{Z}^a_0$. In particular, for $y_0=p^a_i(\xi,\mathcal{E})e^i$ Construction~\ref{const:seller:hedge} gives an optimal superhedging strategy $y$ for the seller.
\end{proposition}

\begin{proof}
We establish that $y\in\Phi$ superhedges $(\xi,\mathcal{E})$ for the seller if and only if $y_t\in\mathcal{Z}^a_t$ for all $t$. Equation \eqref{eq:prop:seller:equivalence-construction:1} then follows directly from \eqref{eq:pi-a}. The minimum in \eqref{eq:prop:seller:equivalence-construction:1} is attained because $\mathcal{Z}_0^a$ is polyhedral, hence closed, and $p^a_i(\xi,\mathcal{E})$ is finite by Proposition~\ref{prop:seller-superhedges-exp}.

If $y\in\Phi$ superhedges $(\xi,\mathcal{E})$ for the seller, then it satisfies \eqref{eq:def:super-hedge:equiv:2}, and clearly $y_T\in\mathcal{Z}^a_T$. For any $t<T$ suppose inductively that $y_{t+1}\in\mathcal{Z}^a_{t+1}$. We have $y_{t+1}\in\mathcal{W}^a_t$ since $y$ is predictable, and $y_t\in\mathcal{V}^a_t$ since it is self-financing. Thus $y_t\in\mathcal{V}^a_t\cap\mathcal{U}^a_t=\mathcal{Z}^a_t$, which concludes the inductive step.

For the converse, fix any $y_0\in\mathcal{Z}^a_0$ and apply Construction \ref{const:seller:hedge}. We now show by induction that the resulting process $y=(y_t)$ satisfies $y_t\in\mathcal{Z}^a_t$ for all~$t$. For any $t\ge0$, suppose by induction that
 $y_t\in\mathcal{L}_{(t-1)\vee0}\cap \mathcal{Z}^a_t$. This means that $y_t\in\mathcal{V}^a_t=\mathcal{W}^a_t+\mathcal{K}_t$, and so $(y_t-\mathcal{K}_t)\cap\mathcal{W}_t\neq\emptyset$. By \eqref{eq:const:seller:hedge} we have both $y_{t+1}\in\mathcal{W}^a_t=\mathcal{Z}^a_{t+1}\cap\mathcal{L}_t$ and $y_t-y_{t+1}\in\mathcal{K}_t$, which concludes the inductive step. The process $y=(y_t)$ that has been constructed is clearly predictable, self-financing and satisfies \eqref{eq:def:super-hedge:equiv:2} since $\mathcal{Z}^a_t\subseteq\mathcal{U}^a_t$ for all~$t$. Thus it superhedges $(\xi,\mathcal{E})$ for the seller, which establishes the correctness of Construction \ref{const:seller:hedge}.

It is now straightforward to see that Construction \ref{const:seller:hedge} with the initial choice $y_0:=p^a_i(\xi,\mathcal{E})e^i\in\mathcal{Z}^a_0$ results in an optimal superhedging strategy $y\in\Phi$ for the seller of $(\xi,\mathcal{E})$.
\end{proof}

Consider now the following result, which will be proved later in this section.

\begin{proposition} \label{prop:construction-of-optimal-martingale-triple}
 There exist $\hat{\chi}\in\mathcal{X}^\mathcal{E}$, $(\hat{\mathbb{Q}},\hat{S})\in\bar{\mathcal{P}}^i(\hat{\chi})$ such that
\begin{equation}
 \mathbb{E}_{\hat{\mathbb{Q}}}((\xi\cdot \hat{S})_{\hat{\chi}}) = -\min\{Z^a_0(s):s\in\sigma_i(\mathbb{R}^d)\},\label{eq:dual-construction:5}
\end{equation}
where $Z^a_0$ is the support function of $-\mathcal{Z}^a_0$.
\end{proposition}

With Proposition~\ref{prop:construction-of-optimal-martingale-triple} in hand, the proof of Theorem \ref{th:seller} is straightforward, so we provide it now.

\begin{proof}[Proof of Theorem \ref{th:seller}]
Note from \eqref{eq:prop:seller:equivalence-construction:1} that
\begin{align*}
p^a_i(\xi,\mathcal{E}) &=
 \min\{x\in\mathbb{R}:xe^i\in\mathcal{Z}^a_0\} \\
&= \min\{x\in\mathbb{R}:-xe^i\in-\mathcal{Z}^a_0\} \\
&= \min\{x\in\mathbb{R}:-xe^i\cdot y \le Z^a_0(y)\text{ for all }y\in\mathbb{R}^d\} \\
&\le \inf\{x\in\mathbb{R}:-xe^i\cdot s \le Z^a_0(s)\text{ for all }s\in\sigma_i(\mathbb{R}^d)\} \\
&= \inf\{x\in\mathbb{R}:-x \le Z^a_0(s)\text{ for all }s\in\sigma_i(\mathbb{R}^d)\} \\
&= -\sup\{x\in\mathbb{R}:x \le Z^a_0(s)\text{ for all }s\in\sigma_i(\mathbb{R}^d)\} \\
&= -\inf\{Z^a_0(s):s\in\sigma_i(\mathbb{R}^d)\} \\
&= -\min\{Z^a_0(s):s\in\sigma_i(\mathbb{R}^d)\} \\
&= \mathbb{E}_{\hat{\mathbb{Q}}}((\xi\cdot \hat{S})_{\hat{\chi}}),
\end{align*}
where the last equality follows from Proposition~\ref{prop:construction-of-optimal-martingale-triple}. Combining this with Proposition~\ref{prop:seller-superhedges-exp} completes the proof.
\end{proof}

Note from the proof above that the randomised stopping time $\hat{\chi}$ and  $\hat{\chi}$-approx\-imate martingale pair $(\hat{\mathbb{Q}},\hat{S})$ of Proposition~\ref{prop:construction-of-optimal-martingale-triple} are optimal for the seller.

The remainder of this section is devoted to establishing Proposition~\ref{prop:construction-of-optimal-martingale-triple}. For all~$t$, let $U^a_t$, $V^a_t$, $W^a_t$, $Z^a_t$ be the support functions of $-\mathcal{U}^a_t$, $-\mathcal{V}^a_t$, $-\mathcal{W}^a_t$, $-\mathcal{Z}^a_t$, respectively.

\begin{remark} \label{remark:construction:when-is-R^d:dual}
Remark \ref{rem:const:seller:when-is-R^d} gives $Z^a_t=V^a_t$ on $\Omega\setminus\mathcal{E}_t$ and $Z^a_t=U^a_t$ on $\Omega\setminus\mathcal{E}^\ast_{t+1}$ for all~$t$. This means that $V^a_t=W^a_t=\delta^\ast_{\mathbb{R}^d}$ on $\Omega\setminus\mathcal{E}^\ast_{t+1}$, whence $Z^a_t=\delta^\ast_{\mathbb{R}^d}$ on $\Omega\setminus\mathcal{E}^\ast_t$ for all~$t$. Moreover, these functions are all polyhedral whenever they are not equal to $\delta^\ast_{\mathbb{R}^d}$ {\cite[Corollary 19.2.1]{rockafellar1996}}.
\end{remark}

Proposition~\ref{prop:construction-of-optimal-martingale-triple} depends on the following technical result.

\begin{lemma} \label{lem:support-technical} {\ }\newline\vspace{-\baselineskip}
 \begin{enumerate}
\item For all $t$ and $y\in\mathcal{L}_t$ we have
\begin{align}
 U^a_t(y) &=
\begin{cases}
 -y\cdot\xi_t&\text{on }\{y\in\mathcal{K}^\ast_t\}\cap\mathcal{E}_t,\\
 0 &\text{on }\{y=0\}\cap(\Omega\setminus\mathcal{E}_t),\\
 \infty &\text{on }\{y\notin\mathcal{K}^\ast_t\}\cup[\{y\neq0\}\cap(\Omega\setminus\mathcal{E}_t)],
\end{cases} \label{eq:U^a_t-support-function-formula}\\
 V^a_t(y) &=
\begin{cases}
 W^a_t(y)&\text{on }\{y\in\mathcal{K}^\ast_t\},\\
 \infty &\text{on }\{y\notin\mathcal{K}^\ast_t\}.
\end{cases} \label{eq:V^a_t-support-function-formula}
\end{align}

\item\label{claim:lem:support-technical:2} Fix any $t$ and $\mu\in\Omega_t$.
\begin{enumerate}
  \item \label{claim:lem:support-technical:2a} If $\mu\subseteq\mathcal{E}_t\cap\mathcal{E}^\ast_{t+1}$, then
\begin{equation} \label{eq:claim:lem:support-technical:2}
 Z^{a\mu}_t = \conv\{U^{a\mu}_t,V^{a\mu}_t\}
\end{equation}
and $\dom Z^{a\mu}_t = \mathcal{K}^{\ast\mu}_t$. Moreover, for each $Y\in\sigma_i(\dom Z^{a\mu}_t)$ there exist $\lambda\in[0,1]$, $X\in \sigma_i(\dom V^{a\mu}_t)$ and $S\in\sigma_i(\dom U^{a\mu}_t)=\sigma_i(\mathcal{K}_t^{\ast\mu})$ such that
\begin{align*}
 Z^{a\mu}_t(Y) &= (1-\lambda)V^{a\mu}_t(X) + \lambda U^{a\mu}_t(S), & Y = (1-\lambda)X + \lambda S.
\end{align*}

  \item \label{claim:lem:support-technical:2b} If $\mu\subseteq\mathcal{E}_t\setminus\mathcal{E}^\ast_{t+1}$, then
\[
 Z^{a\mu}_t=U^{a\mu}_t
\]
and $\dom Z^{a\mu}_t=\mathcal{K}_t^{\ast\mu}$.

  \item \label{claim:lem:support-technical:2c} If $\mu\subseteq\mathcal{E}^\ast_{t+1}\setminus\mathcal{E}_t$, then
\[
 Z^{a\mu}_t = V^{a\mu}_t
\]
and $\dom Z^{a\mu}_t$ is a compactly $i$-generated cone.

\item \label{claim:lem:support-technical:2d} If $\mu\subseteq\Omega\setminus\mathcal{E}^\ast_t$, then
\[
 Z^{a\mu}_t = \delta^\ast_{\mathbb{R}^d}.
\]
\end{enumerate}

  \item\label{claim:lem:support-technical:3} For each $t<T$ and $\mu\in\Omega_t$ with $\mu\subseteq\mathcal{E}^\ast_{t+1}$, we have
\begin{equation} \label{eq:claim:lem:support-technical:3}
 W^{a{\mu}}_t = \conv\{Z^{a\nu}_{t+1}:\nu\in\successors\mu\}
\end{equation}
and $\dom W^{a{\mu}}_t$ is a compactly $i$-generated cone. Moreover, for every $X\in\sigma_i(\dom W^{a{\mu}}_t)$ there exist $p^\nu\ge0$ and $Y^\nu\in\sigma_i(\dom Z^{a\nu}_{t+1})$ for each $\nu\in\successors\mu$ such that
\begin{align*}
 W^{a{\mu}}_t(X) &= \sum_{\nu\in\successors\mu}p^\nu Z^{a\nu}_{t+1}(Y^\nu), &  X &= \sum_{\nu\in\successors\mu}p^\nu Y^\nu, & 1 &= \sum_{\nu\in\successors\mu}p^\nu.
\end{align*}
\end{enumerate}
\end{lemma}

The proof of Lemma \ref{lem:support-technical} is deferred to Appendix \ref{appendix}.

\begin{proof}[Proof of Proposition~\ref{prop:construction-of-optimal-martingale-triple}]
We construct the process $\hat{S}=(\hat{S}_t)$ by backward recursion, together with auxiliary adapted processes $\hat{X}=(\hat{X}_t)$, $\hat{Y}=(\hat{Y}_t)$, $\hat{\lambda}=(\hat{\lambda}_t)$ and predictable $\hat{p}=(\hat{p}_t)$. Fix any $(\mathbb{Q},S)\in\mathcal{P}^i$.

As $\mathcal{E}^\ast_0=\Omega$, Lemma \ref{lem:support-technical}\ref{claim:lem:support-technical:2} ensures that $\sigma_i(\dom Z^a_0)$ is non-empty and compact, and there exists $\hat{Y}_0\in\sigma_i(\dom Z^a_0)$ such that
\begin{equation}
 Z^a_0(\hat{Y}_0) = \min\{Z^a_0(s):s\in\sigma_i(\mathbb{R}^d)\}.
 \label{eq:Z0ofY0}
\end{equation}
Note that $\hat{Y}_0$ is an appropriate starting value for the recursion below since $\Omega\setminus\mathcal{E}^\ast_0=\emptyset$.

For any $t\ge0$, suppose that $\hat{Y}_t$ is an $\mathcal{F}_t$-measurable random variable such that $\hat{Y}_t\in\sigma^i(\dom Z^a_t)$ on $\mathcal{E}_t^\ast$ and $\hat{Y}_t = S_t$ on $\Omega\setminus\mathcal{E}_t^\ast$. For any $\mu\in\Omega_t$ we now construct $\hat{\lambda}_t^\mu\in[0,1]$, $\hat{X}_t\in\sigma_i(\mathcal{K}_t^{\ast})$ and $\hat{S}_t\in\sigma_i(\mathcal{K}_t^{\ast})$ such that
\begin{equation}
 \hat{Y}_t^\mu = (1-\hat{\lambda}_t^\mu)\hat{X}_t^\mu + \hat{\lambda}_t^\mu \hat{S}_t^\mu. \label{eq:dual-construction:6}
\end{equation}
There are four possibilities:
\begin{itemize}
  \item If $\mu\subseteq\mathcal{E}_t\cap\mathcal{E}^\ast_{t+1}$, then Lemma \ref{lem:support-technical}\ref{claim:lem:support-technical:2}\ref{claim:lem:support-technical:2a} ensures the existence of $\hat{\lambda}^\mu_t\in[0,1]$, $\hat{X}_t^\mu\in \sigma_i(\dom V^{a\mu}_t)$ and $\hat{S}_t^\mu\in\sigma_i(\mathcal{K}_t^{\ast\mu})$ satisfying \eqref{eq:dual-construction:6} and
\begin{equation}
 Z^{a\mu}_t(\hat{Y}_t^\mu) = (1-\hat{\lambda}_t^\mu)V^{a\mu}_t(\hat{X}_t^\mu) + \hat{\lambda}_t^\mu U^a_t(\hat{S}_t^\mu).  \label{eq:dual-construction:71}
\end{equation}
This possibility does not arise when $t=T$ because $\mathcal{E}_{T+1}^\ast=\emptyset$.

  \item If $\mu\subseteq\mathcal{E}_t\setminus\mathcal{E}^\ast_{t+1}$, then Lemma \ref{lem:support-technical}\ref{claim:lem:support-technical:2}\ref{claim:lem:support-technical:2b} applies. Choosing $\hat{\lambda}^\mu_t:=1$, $\hat{X}_t^\mu:=S_t^\mu$ and $\hat{S}_t^\mu:=\hat{Y}_t^\mu$ yields \eqref{eq:dual-construction:6} and
\begin{equation}
 Z^{a\mu}_t(\hat{Y}_t^\mu) = U^a_t(\hat{S}_t^\mu) = \hat{\lambda}_t^\mu U^a_t(\hat{S}_t^\mu).  \label{eq:dual-construction:72}
\end{equation}

  \item If $\mu\subseteq\mathcal{E}^\ast_{t+1}\setminus\mathcal{E}_t$, then Lemma \ref{lem:support-technical}\ref{claim:lem:support-technical:2}\ref{claim:lem:support-technical:2c} gives \eqref{eq:dual-construction:6} and
\begin{equation}
 Z^{a\mu}_t(\hat{Y}_t^\mu) = V^{a\mu}_t(\hat{X}_t^\mu) = (1-\hat{\lambda}_t^\mu)V^{a\mu}_t(\hat{X}_t^\mu)  \label{eq:dual-construction:73}
\end{equation}
after defining $\hat{\lambda}^\mu_t:=0$, $\hat{X}_t^\mu:=\hat{Y}_t^\mu$ and $\hat{S}_t^\mu:=S_t^\mu$. This possibility does not arise when $t=T$ because $\mathcal{E}_{T+1}^\ast=\emptyset$.

\item If $\mu\not\subseteq\mathcal{E}^\ast_t$, then $\hat{Y}_t^\mu=S_t^\mu$ by the recursive assumption, and \[Z^{a\mu}_t=V^{a\mu}_t=U^{a\mu}_t=\delta^\ast_{\mathbb{R}^d}\]
by Remark \ref{remark:construction:when-is-R^d:dual} and Lemma \ref{lem:support-technical}\ref{claim:lem:support-technical:2}\ref{claim:lem:support-technical:2d}. Defining $\lambda^\mu_t:=0$ and $\hat{S}_t^\mu:=\hat{X}_t^\mu:=S_t^\mu$ gives \eqref{eq:dual-construction:6}. This possibility does not arise when $t=0$ because $\mathcal{E}_0^\ast=\Omega$.
\end{itemize}

Note that $\hat{X}_t\in\sigma_i(\dom V^a_t)$ on $\mathcal{E}^\ast_{t+1}$ and $\hat{X}_t=S_t$ on $\Omega\setminus\mathcal{E}^\ast_{t+1}$. For any $t<T$ and $\mu\in\Omega_t$ we now construct $(\hat{Y}^\nu_t)_{\nu\in\successors\mu}$ and $(\hat{p}_{t+1}^\nu)_{\nu\in\successors\mu}$ such that
\begin{align}
 1 &= \sum_{\nu\in\successors\mu}p_{t+1}^\nu,\label{eq:dual-construction:8}\\
  \hat{X}_t^\mu &= \sum_{\nu\in\successors\mu}p_{t+1}^\nu \hat{Y}_{t+1}^\nu,  \label{eq:dual-construction:3}
\end{align}
There are two possibilities:
\begin{itemize}
 \item If $\mu\subseteq\mathcal{E}^\ast_{t+1}$, then $\hat{X}_t\in\sigma_i(\dom V^a_t)$ and Lemma \ref{lem:support-technical}\ref{claim:lem:support-technical:3} assures the existence of  $(\hat{Y}^\nu_t)_{\nu\in\successors\mu}$ and $(\hat{p}_{t+1}^\nu)_{\nu\in\successors\mu}$ with $\hat{p}_{t+1}^\nu\in[0,1]$, $\hat{Y}_{t+1}^\nu\in\sigma_i(\dom Z^{a\nu}_{t+1})$ for all $\nu\in\successors\mu$ satisfying \eqref{eq:dual-construction:8}--\eqref{eq:dual-construction:3} and
\begin{equation}
  V^{a\mu}_t(\hat{X}_t^\mu) = W^{a\nu}_t(\hat{X}_t^\mu) = \sum_{\nu\in\successors\mu}\hat{p}_{t+1}^\nu Z^{a\nu}_{t+1}(\hat{Y}_{t+1}^\nu).  \label{eq:dual-construction:4}
\end{equation}

  \item If $\mu\not\subseteq\mathcal{E}^\ast_{t+1}$, then defining $\hat{Y}_{t+1}^\nu:=S^\nu_{t+1}$ and $\hat{p}_{t+1}^\nu:=\mathbb{Q}(\nu|\mu)$ for all $\nu\in\successors\mu$ gives \eqref{eq:dual-construction:8}--\eqref{eq:dual-construction:3}.
\end{itemize}
This concludes the recursive step.

The probability measure $\hat{\mathbb{Q}}$ is defined as
\[
 \hat{\mathbb{Q}}(\omega) := \prod_{t=0}^{T-1}\hat{p}_{t+1}^\omega \text{ for all } \omega.
\]
Then \eqref{eq:dual-construction:8}--\eqref{eq:dual-construction:4} gives
\begin{align}
 V^a_t(\hat{X}_t) &= \mathbb{E}_{\hat{\mathbb{Q}}}(Z^a_{t+1}(\hat{Y}_{t+1})|\mathcal{F}_t) \text{ on } \mathcal{E}^\ast_{t+1}, & \hat{X}_t &= \mathbb{E}_{\hat{\mathbb{Q}}}(\hat{Y}_{t+1}|\mathcal{F}_t) \label{eq:dual-construction:84}
\end{align}
for all $t<T$.

The randomised stopping time $\hat{\chi}=(\hat{\chi}_t)$ is defined by $\hat{\chi}_0:=\hat{\lambda}_0$ and
\[
 \hat{\chi}_t := \hat{\lambda}_t\left[1-\sum_{s=0}^{t-1}\hat{\chi}_s\right]
\]
for $t>0$. It is clear from the construction that $\hat\lambda_t=0$ on $\Omega\setminus\mathcal{E}_t$ for all $t$, which implies that $\hat{\chi}\in\mathcal{X}^\mathcal{E}$. Observe also that
\begin{align*}
\hat{\lambda}_t\hat{\chi}_t^\ast &= \hat{\chi}_t, & (1-\hat{\lambda}_t)\hat{\chi}_t^\ast &= \hat{\chi}_t^\ast - \chi_t = \hat{\chi}_{t+1}^\ast
\end{align*}
for all $t$; recall that $\hat{\chi}_{T+1}^\ast=0$ by definition. It follows from \eqref{eq:dual-construction:6} that
\[
 \hat{\chi}_t^\ast \hat{Y}_t = \hat{\chi}_{t+1}^\ast \hat{X}_t + \hat{\chi}_t \hat{S}_t
\]
for all $t$. Equations \eqref{eq:U^a_t-support-function-formula} and \eqref{eq:dual-construction:72} give
\begin{equation} \label{eq:dual-construction:72-cons}
 \hat{\chi}_t^\ast Z^a_t(\hat{Y}_t) = \hat{\chi}_tU^a_t(\hat{S}_t) = - \hat{\chi}_t \xi_t\cdot\hat{S}_t \text{ on }\mathcal{E}_t\setminus\mathcal{E}^\ast_{t+1}.
\end{equation}
Since $\hat{\chi}_t=0$ on $\Omega\setminus\mathcal{E}_t$, equations \eqref{eq:dual-construction:71} and \eqref{eq:dual-construction:73} may be combined with \eqref{eq:U^a_t-support-function-formula} to yield
\begin{equation} \label{eq:dual-construction:713}
 \hat{\chi}_t^\ast Z^a_t(\hat{Y}_t) = \hat{\chi}_{t+1}^\ast V^a_t(\hat{X}_t) - \hat{\chi}_t \xi_t\cdot\hat{S}_t \text{ on }\mathcal{E}^\ast_{t+1}.
\end{equation}

It is possible to show by backward induction that
\begin{equation}
 \mathbb{E}_{\hat{\mathbb{Q}}}(\hat{S}_{t+1}^{\hat{\chi}\ast}|\mathcal{F}_t) = \hat{\chi}_{t+1}^\ast \hat{X}_t \label{eq:dual-construction:1}
\end{equation}
for all $t$. At time $t=T$ this follows from the notational conventions $\hat{S}^{\hat{\chi}\ast}_{T+1}=0$ and $\hat{\chi}^\ast_{T+1}=0$. Suppose that \eqref{eq:dual-construction:1} holds for some $t>0$. Then
\begin{align*}
 \mathbb{E}_{\hat{\mathbb{Q}}}(\hat{S}_t^{\hat{\chi}\ast}|\mathcal{F}_{t-1})
&= \mathbb{E}_{\hat{\mathbb{Q}}}(\hat{\chi}_t\hat{S}_t + \hat{S}_{t+1}^{\hat{\chi}\ast}|\mathcal{F}_{t-1}) \\
&= \mathbb{E}_{\hat{\mathbb{Q}}}(\hat{\chi}_t\hat{S}_t + \mathbb{E}_{\hat{\mathbb{Q}}}(\hat{S}_{t+1}^{\hat{\chi}\ast}|\mathcal{F}_t)|\mathcal{F}_{t-1}) \\
&= \mathbb{E}_{\hat{\mathbb{Q}}}(\hat{\chi}_t\hat{S}_t + \hat{\chi}_{t+1}^\ast \hat{X}_t|\mathcal{F}_{t-1}) \\
&= \hat{\chi}^\ast_t\mathbb{E}_{\hat{\mathbb{Q}}}(\hat{Y}_t|\mathcal{F}_{t-1}) = \hat{\chi}_t^\ast \hat{X}_{t-1},
\end{align*}
which concludes the inductive step.

We also show by backward induction that
\begin{equation} \label{eq:dual-construction:2}
  \hat{\chi}_t^\ast Z^a_t(\hat{Y}_t) = -\mathbb{E}_{\hat{\mathbb{Q}}}((\xi\cdot\hat{S})_t^{\hat{\chi}\ast}|\mathcal{F}_t) \text{ on } \mathcal{E}^\ast_t
\end{equation}
for all $t$. If $t=T$ then $\mathcal{E}^\ast_{T+1}=\emptyset$ and \eqref{eq:dual-construction:72-cons} gives
\[
 \hat{\chi}_T^\ast Z^a_T(\hat{Y}_T) = -\hat{\chi}_T\xi_T\cdot\hat{S}_T = -(\xi\cdot S)^{\hat{\chi}\ast}_T = -\mathbb{E}_{\hat{\mathbb{Q}}}((\xi\cdot S)^{\hat{\chi}\ast}_T|\mathcal{F}_T)
\]
on $\mathcal{E}_T=\mathcal{E}^\ast_T$. Suppose now that
\[
  \hat{\chi}_{t+1}^\ast Z^a_{t+1}(\hat{Y}_{t+1}) = -\mathbb{E}_{\hat{\mathbb{Q}}}((\xi\cdot\hat{S})_{t+1}^{\hat{\chi}\ast}|\mathcal{F}_{t+1}) \text{ on } \mathcal{E}^\ast_{t+1}
\]
 holds for some $t<T$. On the set $\mathcal{E}_t\setminus\mathcal{E}^\ast_{t+1}$ we have $\hat{\chi}^\ast_{t+1}=0$, whence $(\xi\cdot S)^{\hat{\chi}\ast}_{t+1}=0$. Equation \eqref{eq:dual-construction:72-cons} then gives
\[
 \hat{\chi}_t^\ast Z^a_t(\hat{Y}_t)
= - \hat{\chi}_t\xi_t\cdot\hat{S}_t - \mathbb{E}_{\hat{\mathbb{Q}}}((\xi\cdot S)^{\hat{\chi}\ast}_{t+1}|\mathcal{F}_t) \\
= -\mathbb{E}_{\hat{\mathbb{Q}}}((\xi\cdot S)^{\hat{\chi}\ast}_t|\mathcal{F}_t).
\]
On $\mathcal{E}_{t+1}^\ast$ the equations \eqref{eq:dual-construction:84} and \eqref{eq:dual-construction:713} give
\begin{align*}
 \hat{\chi}_t^\ast Z^a_t(\hat{Y}_t)
&= \hat{\chi}_{t+1}^\ast V^a_t(\hat{X}_t) - \hat{\chi}_t \xi_t\cdot\hat{S}_t \\
&= \hat{\chi}_{t+1}^\ast \mathbb{E}_{\hat{\mathbb{Q}}}(Z^a_{t+1}(\hat{Y}_{t+1})|\mathcal{F}_t) - \hat{\chi}_t \xi_t\cdot\hat{S}_t \\
&= - \mathbb{E}_{\hat{\mathbb{Q}}}((\xi\cdot\hat{S})_{t+1}^{\hat{\chi}\ast}|\mathcal{F}_t) - \hat{\chi}_t \xi_t\cdot\hat{S}_t \\
&= - \mathbb{E}_{\hat{\mathbb{Q}}}((\xi\cdot\hat{S})_t^{\hat{\chi}\ast}|\mathcal{F}_t).
\end{align*}
This concludes the inductive step since $\mathcal{E}^\ast_t=[\mathcal{E}_t\setminus\mathcal{E}^\ast_{t+1}]\cup\mathcal{E}_{t+1}^\ast$.

To summarize, we have constructed a randomised stopping time $\hat{\chi}\in\mathcal{X}^\mathcal{E}$, a probability measure $\hat{\mathbb{Q}}$ and an adapted process $\hat{S}$ such that $\hat{S}_t\in\sigma_i(\mathcal{K}_t^\ast)$ and
\[
 \mathbb{E}_{\hat{\mathbb{Q}}}(S^{\hat{\chi}\ast}_{t+1}|\mathcal{F}_t) = \hat{\chi}^\ast_{t+1} \hat{X}_t \in\dom V^a_t\subseteq\mathcal{K}^\ast_t
\]
for all $t$. Equation \eqref{eq:dual-construction:2} moreover gives
\begin{align*}
  \mathbb{E}_{\hat{\mathbb{Q}}}((\xi\cdot\hat{S})_{\hat{\chi}}) = -Z^a_0(\hat{Y}_0)
\end{align*}
which together with \eqref{eq:Z0ofY0} leads to \eqref{eq:dual-construction:5} and completes the proof of Proposition~\ref{prop:construction-of-optimal-martingale-triple}.
\end{proof}

\begin{example}
\label{Exl:2}\upshape Consider a single-step model with four nodes at
time~$1$, that is, $\Omega=\{\omega_{1},\omega_{2},\omega_{3},\omega
_{4}\}$, and three assets. We take asset~3 to be a cash account with zero
interest rate, take the cash prices of assets 1 and 2 in a friction-free
market in Table~\ref{tab:1}, and introduce transaction costs at the rate $k=\frac{1}{6}$ in a similar
manner as in Example~\ref{Exl:1}, with the matrix-valued exchange rate process%
\[
\pi_{t}=\left[
\begin{array}
[c]{ccc}%
1 & (1+k)S_{t}^{2}/S_{t}^{1} & (1+k)/S_{t}^{1}\\
(1+k)S_{t}^{1}/S_{t}^{2} & 1 & (1+k)/S_{t}^{2}\\
(1+k)S_{t}^{1} & (1+k)S_{t}^{2} & 1
\end{array}
\right]
\]
for $t=0,1$. Consider the American option with payoff process $\xi$ in Table \ref{tab:1} in this model; its exercise policy is $\mathcal{E}_0=\mathcal{E}_1=\Omega$.

\begin{table}
\caption{Friction-free cash prices and American option payoff, Example~\ref{Exl:2}}
\label{tab:1}       
\begin{tabular}{ccccccc}
\hline\noalign{\smallskip}
& $S_{0}^{1}$ & $S_{0}^{2}$ & $S_{1}^{1}$ & $S_{1}^{2}$ & $\xi_{0}=(\xi_{0}^{1},\xi_{0}^{2},\xi_{0}^{3})$ & $\xi_{1}=(\xi_{1}^{1}%
,\xi_{1}^{2},\xi_{1}^{3})$ \\
\noalign{\smallskip}\hline\noalign{\smallskip}
$\omega_{1}$ & $10$ & $20$ & $\phantom18$ & $18$ & $(1,-1,33)$ & $(-1,1,10)$\\
$\omega_{2}$ & $10$ & $20$ & $12$ & $18$ & $(1,-1,33)$ & $(-2,1,10)$\\
$\omega_{3}$ & $10$ & $20$ & $\phantom18$ & $22$ & $(1,-1,33)$ & $(-1,2,10)$\\
$\omega_{4}$ & $10$ & $20$ & $12$ & $22$ & $(1,-1,33)$ & $(-2,2,10)$ \\
\noalign{\smallskip}\hline
\end{tabular}
\end{table}

Construction~\ref{const:seller} is
formulated in terms of the convex sets $\mathcal{U}_{t}^{a}$, $\mathcal{V}%
_{t}^{a}$, $\mathcal{W}_{t}^{a}$, $\mathcal{Z}_{t}^{a}$, but it is easier to
visualise it by drawing the support functions $U_{t}^{a}$, $V_{t}^{a}$, $W_{t}%
^{a}$, $Z_{t}^{a}$ of $-\mathcal{U}_{t}^{a}$, $-\mathcal{V}_{t}^{a}$, $-\mathcal{W}_{t}^{a}$, $-\mathcal{Z}_{t}^{a}$ or indeed the sections
$-\sigma_{3}(\epi U_{t}^{a})$, $-\sigma_{3}(\epi V_{t}^{a})$, $-\sigma
_{3}(\epi W_{t}^{a})$, $-\sigma_{3}(\epi Z_{t}^{a})$, which are shown in
Figure~\ref{algorithm_seller.png} for the above single-step model with transaction costs.
Observe that all the polyhedra in Figure~\ref{algorithm_seller.png} are
unbounded below, but have been truncated when drawing the pictures.

The construction proceeds as follows:
\begin{itemize}
\item We start with $-\sigma_{3}(\epi U_{1}^{a})=-\sigma_{3}(\epi Z_{1}^{a})$
for all four nodes at time~$1$, represented by the four dark gray polyhedra in
Figure~\ref{algorithm_seller.png}(a). These are computed using~\eqref{eq:U^a_t-support-function-formula} in
Lemma~\ref{lem:support-technical}.

\item We then take the convex hull of these four polyhedra to obtain
$-\sigma_{3}(\epi W_{0}^{a})$, the semi-transparent gray polyhedron in
Figures~\ref{algorithm_seller.png}(a), (b), (c), (d). Formula~\eqref{eq:claim:lem:support-technical:3} in
Lemma~\ref{lem:support-technical} is used here.

\item Next, $-\sigma_{3}(\epi V_{0}^{a})$, the dark gray polyhedron in
Figure~\ref{algorithm_seller.png}(b), is the intersection of $-\sigma
_{3}(\epi
W_{0}^{a})$ and $-\sigma_{3}(\mathcal{K}_{0}^{\ast})$, according~\eqref{eq:V^a_t-support-function-formula} in
Lemma~\ref{lem:support-technical}.

\item Then we take $-\sigma_{3}(\epi U_{0}^{a})$, the dark gray polyhedron in
Figure~\ref{algorithm_seller.png}(c). This is computed using~\eqref{eq:U^a_t-support-function-formula} in
Lemma~\ref{lem:support-technical}.

\item Finally, we obtain $-\sigma_{3}(\epi Z_{0}^{a})$, the dark gray
polyhedron in Figure~\ref{algorithm_seller.png}(d), as the convex hull of
$-\sigma_{3}(\epi V_{0}^{a})$ and $-\sigma_{3}(\epi U_{0}^{a})$, according to
\eqref{eq:claim:lem:support-technical:2} in Lemma~\ref{lem:support-technical}.
\end{itemize}

\begin{figure}[th]
\includegraphics[
natheight=1.179600in,
natwidth=4.733100in,
height=1.2142in,
width=4.7867in
]%
{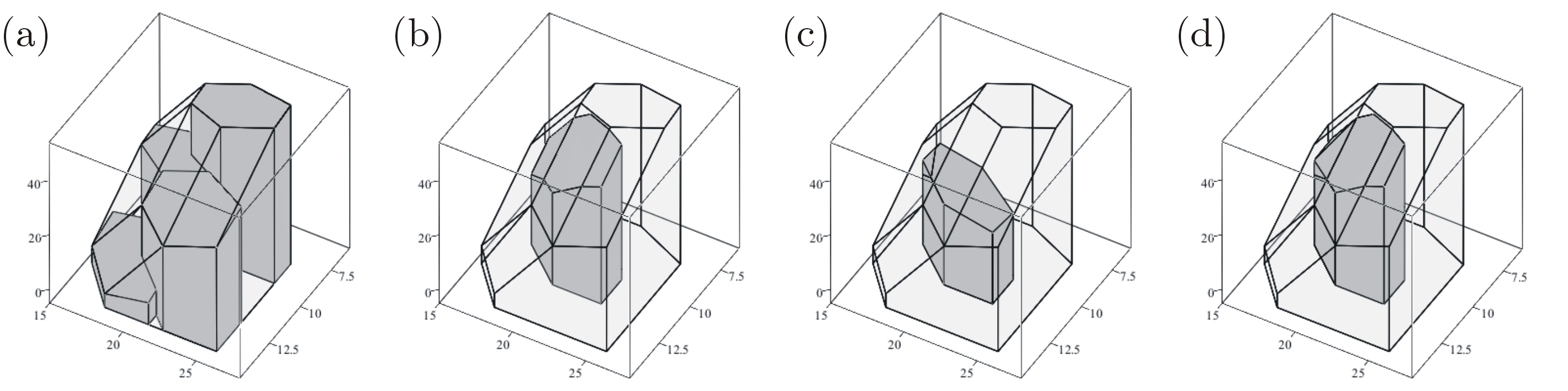}%
\caption{Construction~\ref{const:seller} expressed in terms of $-\sigma_{3}(\epi U_0^{a})$, $-\sigma_{3}%
(\epi V_0^{a})$, $-\sigma_{3}(\epi
W_0^{a})$, $-\sigma_{3}(\epi Z_0^{a})$, Example~\ref{Exl:2}}%
\label{algorithm_seller.png}%
\end{figure}

\noindent The ask price of the American option is the maximum of $-\sigma_{3}(\epi
Z_{0}^{a})$; see Theorem~\ref{th:seller}. The polyhedron $-\sigma_{3}(\epi Z_{0}^{a})$ has
$10$ vertices:\begin{gather*}
(10,120/7,181/7),(60/7,132/7,262/7),(35/3,22,106/3),(35/3,70/3,38),\\
(60/7,120/7,184/7),(35/3,20,950/33),(11,132/7,194/7),(66/7,22,310/7),\\
(60/7,20,3170/77),(10,70/3,134/3),
\end{gather*}
and its highest point turns out to be at $\frac{134}{3}%
\cong44.67$. This is the ask price of the American option.
\end{example}

\section{Pricing and hedging for the buyer}
\label{sec:buyer}

Recall that a pair $(y,\tau)$ consisting of a self-financing trading strategy $y\in\Phi$ and stopping time $\tau\in\mathcal{T}^\mathcal{E}$ superhedges the option $(\xi,\mathcal{E})$ for the buyer if \eqref{eq:def:super-hedge:buyer} holds, equivalently if $y_\tau\in\mathcal{U}^b_\tau$.

The next result shows that the set $\mathcal{Z}_0^b$ given by Construction \ref{const:buyer} is the collection of initial endowments allowing the buyer to superhedge $(\xi,\mathcal{E})$, and that it can be used to compute the bid price directly.

\begin{proposition} \label{prop:buyer:equivalence-construction}
We have
\begin{multline*}
\begin{aligned}
 \mathcal{Z}_0^b &= \{y_0\in\mathbb{R}^d:(y,\tau)\in\Phi\times\mathcal{T}^\mathcal{E} \text{ superhedges } (\xi,\mathcal{E}) \text{ for the buyer}\} \\
 &= \{y_0\in\mathbb{R}^d:(y,\tau)\in\Phi\times\mathcal{T}^\mathcal{E} \text{ superhedges } (\xi,\mathcal{E}) \text{ for the buyer} \\
\end{aligned}\\
 \text{ and }\tau = \min\{t:y_t\in\mathcal{U}^b_t\}\}
\end{multline*}
and
\begin{equation} \label{eq:prop:buyer:equivalence-construction:1}
 p^b_i(\xi,\mathcal{E}) = -\min\{x\in\mathbb{R}|xe^i\in\mathcal{Z}^b_0\}.
\end{equation}
{Moreover, Construction \ref{const:buyer:hedge} yields a superhedging strategy $(y,\tau_T)\in\Phi\times\mathcal{T}$ for the buyer for any initial endowment $y_0\in\mathcal{Z}^b_0$. In particular, for $y_0=-p^b_i(\xi,\mathcal{E})e^i$ Construction~\ref{const:buyer:hedge} yields an optimal superhedging strategy $(y,\tau_T)$ for the buyer.}
\end{proposition}

\begin{proof}
We show below that $z\in\mathcal{Z}^b_0$ if and only if there exists a superhedging strategy $(y,\tau)$ for the buyer of $(\xi,\mathcal{E})$ with $y_0=z$ and
\begin{equation} \label{eq:prop:buyer:equivalence-construction:3}
 \tau=\min\{t:y_t\in\mathcal{U}^b_t\}.
\end{equation}
The two representations of $\mathcal{Z}^b_0$ are equivalent since if $(y,\tau)$ superhedges $(\xi,\mathcal{E})$ for the buyer and
\[
 \tau':=\min\{t:y_t\in\mathcal{U}^b_t\},
\]
then $(y,\tau')$ also superhedges $(\xi,\mathcal{E})$ for the buyer. Once the result for $\mathcal{Z}^b_0$ is established, equation \eqref{eq:prop:buyer:equivalence-construction:1} follows directly from \eqref{eq:pi-b}. The minimum is attained because~$\mathcal{Z}^b_0$ is closed and $p^b_i(\xi,\mathcal{E})$ is finite.

Suppose that $(y,\tau)\in\Phi\times\mathcal{T}^\mathcal{E}$ superhedges $(\xi,\mathcal{E})$ for the buyer and satisfies~\eqref{eq:prop:buyer:equivalence-construction:3}. We show by backward induction on $t$ that $y_t\in\mathcal{Z}^b_t\setminus\mathcal{U}^b_t$ on $\{\tau>t\}$ for all~$t$. At time $t=T$ this is trivial because $\{\tau>T\}=\emptyset$. For any $t<T$, suppose that $y_{t+1}\in\mathcal{Z}^b_{t+1}\setminus\mathcal{U}^b_{t+1}$ on $\{\tau> t+1\}$. Since $y_{t+1}\in\mathcal{U}^b_{t+1}$ on $\{\tau=t+1\}$, this means that $y_{t+1}\in\mathcal{Z}^b_{t+1}$ on $\{\tau> t\} = \{\tau\ge t+1\}$.
On $\{\tau>t\}$ we then have $y_{t+1}\in\mathcal{W}^b_t$ as $y_{t+1}\in\mathcal{L}_t$, and $y_t\in\mathcal{V}^b_t\subseteq\mathcal{Z}^b_t$ because of the self-financing property. However $y_t\notin\mathcal{U}^b_t$ on $\{\tau> t\}$ because of  \eqref{eq:prop:buyer:equivalence-construction:3}, and so $y_t\in\mathcal{Z}^b_t\setminus\mathcal{U}^b_t$ on $\{\tau>t\}$, which concludes the inductive step. Finally, $y_0\in\mathcal{U}^b_0$ if $\tau=0$ and $y_0\in\mathcal{Z}^b_0\setminus\mathcal{U}^b_0$ if $\tau>0$, and therefore $y_0\in\mathcal{Z}^b_0$.

Conversely, we can use Construction~\ref{const:buyer:hedge} to produce
sequences $y=(y_{t})$ and $(\tau_{t})$ from any initial point $y_{0}%
:=z\in\mathcal{Z}_{0}^{b}$. We shall verify that $y$ is a predictable process
and $\tau_{T}$ is a stopping time. For any $t<T$ suppose by induction that
$y_{t}$ is $\mathcal{F}_{\left(  t-1\right)  \vee0}$-measurable and $\tau_{t}$
is a stopping time (and observe that for $t=0$ these conditions are
satisfied). Then $\left\{  \tau_{t}\leq t\right\}  \in\mathcal{F}_{t}$ and
$y_{t}$~is $\mathcal{F}_{t}$-measurable, which implies that $y_{t+1}%
\mathbf{1}_{\left\{  \tau_{t}\leq t\right\}  }=y_{t}\mathbf{1}_{\left\{
\tau_{t}\leq t\right\}  }$ is $\mathcal{F}_{t}$-measurable. We also have
$y_{t+1}\mathbf{1}_{\left\{  \tau_{t}=t+1\right\}  }\in\mathcal{W}_{t}%
^{b}\subseteq\mathcal{L}_{t}$ by~\eqref{eq:const:buyer:hedge:1}. It follows
that $y_{t+1}=y_{t}\mathbf{1}_{\left\{  \tau_{t}\leq t\right\}  }%
+y_{t+1}\mathbf{1}_{\left\{  \tau_{t}=t+1\right\}  }$ is $\mathcal{F}_{t}%
$-measurable. To show that $\tau_{t+1}$ is a stopping time, we need to verify
that $\left\{  \tau_{t+1}\leq s\right\}  \in\mathcal{F}_{s}$ for each$~s$. For
any $s\leq t$ this is satisfied because $\left\{  \tau_{t+1}\leq s\right\}
=\left\{  \tau_{t}\leq s\right\}  \in\mathcal{F}_{s}$, and for any $s>t+1$ we
have $\left\{  \tau_{t+1}\leq s\right\}  =\Omega\in\mathcal{F}_{s}$. It
remains to check for $s=t+1$ that $\left\{  \tau_{t+1}\leq t+1\right\}
\in\mathcal{F}_{t+1}$. We have $\left\{  \tau_{t+1}\leq t+1\right\}  =\left\{
\tau_{t}\leq t\right\}  \cup\left\{  \tau_{t+1}=t+1\right\}  $ and because
$\left\{  \tau_{t}\leq t\right\}  \in\mathcal{F}_{t}$, we only need to observe
that $\left\{  \tau_{t+1}=t+1\right\}  =\left\{  \tau_{t}=t+1\right\}
\cap\left\{  y_{t+1}\in\mathcal{U}_{t+1}^{b}\right\}  $ belongs to
$\mathcal{F}_{t+1}$. This is so because $\left\{  \tau_{t}=t+1\right\}
=\left\{  \tau_{t}>t\right\}  \in\mathcal{F}_{t}\subseteq\mathcal{F}_{t+1}$
and $\left\{  y_{t+1}\in\mathcal{U}_{t+1}^{b}\right\}  \in\mathcal{F}_{t+1}$
since $y_{t+1}$ has already been shown to be $\mathcal{F}_{t}$-measurable and
therefore $\mathcal{F}_{t+1}$-measurable. This completes the induction
argument. Moreover, observe from~\eqref{eq:const:buyer:hedge:1} that
$y_{t}-y_{t+1}\in\mathcal{K}_{t}$ for all $t<T$, that is, $y$ is a
self-financing strategy. Combined with predictability, it means that~$y\in
\Phi$. Furthermore, observe that $\tau_{T}\leq T$. This is so
by~\eqref{eq:const:buyer:hedge:1} since $\left\{  \tau_{T}>T\right\}
=\left\{  \tau_{T}=T+1\right\}  \subseteq\left\{  y_{T}\in\mathcal{Z}_{T}%
^{b}\setminus\mathcal{U}_{T}^{b}\right\}  =\emptyset$ given that
$\mathcal{Z}_{T}^{b}=\mathcal{U}_{T}^{b}$. Since we already know that
$\tau_{T}$ is a stopping time, we can conclude that $\tau_{T}\in\mathcal{T}$.

Next we show by induction that $y_{\tau_{t}}\in\mathcal{U}_{\tau_{t}}^{b}$ on
$\left\{  \tau_{t}\leq t\right\}  $ for all~$t$. This is clearly so for $t=0$.
Suppose that $y_{\tau_{t}}\in\mathcal{U}_{\tau_{t}}^{b}$ on $\left\{  \tau
_{t}\leq t\right\}  $ for some $t<T$. By~\eqref{eq:const:buyer:hedge:1}, on
$\left\{  \tau_{t}\leq t\right\}  $ we have $\tau_{t}=\tau_{t+1}$, and so
$y_{\tau_{t+1}}\in\mathcal{U}_{\tau_{t+1}}^{b}$ by the induction hypothesis.
Moreover, on $\left\{  \tau_{t+1}=t+1\right\}  $ we have $y_{t+1}%
\in\mathcal{U}_{t+1}^{b}$, which can be written as $y_{\tau_{t+1}}%
\in\mathcal{U}_{\tau_{t+1}}^{b}$. This shows that $y_{\tau_{t+1}}%
\in\mathcal{U}_{\tau_{t+1}}^{b}$ on $\left\{  \tau_{t+1}\leq t+1\right\}
=\left\{  \tau_{t}\leq t\right\}  \cup\left\{  \tau_{t+1}=t+1\right\}  $,
completing the induction step. In particular, it follows that $y_{\tau_{T}}%
\in\mathcal{U}_{\tau_{T}}^{b}$.

Finally, we shall see that $\tau_{T}\in\mathcal{T}^{\mathcal{E}}$. We already
know that $\tau_{T}\in\mathcal{T}$. We also know that $y_{\tau_{t}}%
\in\mathcal{U}_{\tau_{t}}^{b}$ on $\left\{  \tau_{t}\leq t\right\}  $, so%
\[
\left\{  \tau_{T}=t\right\}  \subseteq\left\{  y_{t}\in\mathcal{U}_{t}%
^{b}\right\}  \subseteq\left\{  \mathcal{U}_{t}^{b}\neq\emptyset\right\}
=\mathcal{E}_{t}%
\]
for any~$t$. It means that $\tau_{T}\in\mathcal{T}^{\mathcal{E}}$.

We have verified that that $(y,\tau_{T})$ is a superhedging strategy for the
buyer of $\left(  \xi,\mathcal{E}\right)  $. To complete the proof observe
that on $\left\{  y_{0}\notin\mathcal{U}_{0}^{b},\ldots,y_{t-1}\notin
\mathcal{U}_{t-1}^{b},y_{t}\in\mathcal{U}_{t}^{b}\right\}  $ we have $\tau
_{0}=1,\ldots,\tau_{t-1}=t$ and $\tau_{t}=\cdots=\tau_{T}=t$
by~\eqref{eq:const:buyer:hedge:2}, which implies that $\tau_{T}=\min\left\{
t:y_{t}\in\mathcal{U}_{t}^{b}\right\}  $.
\end{proof}

Let us now establish Theorem \ref{th:buyer}.

\begin{proof}[Proof of Theorem \ref{th:buyer}]
Note first that $(y,\tau)$ is a superhedging strategy for the buyer of $(\xi,\mathcal{E})$ if and only if it is a superhedging strategy for the seller of the European-style option with payoff $-\xi_\tau$ and expiration date $\tau$ of Section \ref{sec:European}. Denoting the European-style option by $(-\xi_\tau,\tau)$, the bid price of $(\xi,\mathcal{E})$ defined in \eqref{eq:pi-b} can be written as
\begin{align}
 p^b_i(\xi,\mathcal{E})
&= \max_{\tau\in\mathcal{T}^\mathcal{E}}\sup\{-x\in\mathbb{R}:\exists y\in\Phi\text{ with }y_0=xe^i\text{ such that }y_\tau+\xi_\tau\in\mathcal{K}_\tau\} \nonumber \\
&= \max_{\tau\in\mathcal{T}^\mathcal{E}}[-\inf\{x\in\mathbb{R}:\exists y\in\Phi\text{ with }y_0=xe^i\text{ such that }y_\tau+\xi_\tau\in\mathcal{K}_\tau\}] \nonumber\\
&= \max_{\tau\in\mathcal{T}^\mathcal{E}}[-p^a_i(-\xi_\tau,\tau)]. \label{eq:th:buyer:2}
\end{align}
The equality \eqref{eq:th:buyer:2} shows that $p^b_i(\xi,\mathcal{E})$ is finite because $\mathcal{T}^\mathcal{E}$ is finite and the ask prices are all finite by Proposition~\ref{prop:seller-superhedges-exp}. Equation \eqref{eq:ask:E} in conjunction with Lemma~\ref{lem:closure} {gives}
\begin{equation} \label{eq:th:buyer:1}
 -p^a_i(-\xi_\tau,\tau)= \inf_{(\mathbb{Q},S)\in\mathcal{P}^i(\tau)} \mathbb{E}_\mathbb{Q}((\xi\cdot S)_\tau)  =  \min_{(\mathbb{Q},S)\in\bar{\mathcal{P}}^i(\tau)} \mathbb{E}_\mathbb{Q}((\xi\cdot S)_\tau),
\end{equation}
and so
\[
 p^b_i(\xi,\mathcal{E})
= \max_{\tau\in\mathcal{T}^\mathcal{E}}\inf_{(\mathbb{Q},S)\in\mathcal{P}^i(\tau)} \mathbb{E}_\mathbb{Q}((\xi\cdot S)_\tau)
= \max_{\tau\in\mathcal{T}^\mathcal{E}}\min_{(\mathbb{Q},S)\in\bar{\mathcal{P}}^i(\tau)} \mathbb{E}_\mathbb{Q}((\xi\cdot S)_\tau).
\]

An optimal superhedging strategy $(\check{y},\check{\tau})$ for the buyer of $(\xi,\mathcal{E})$ may be constructed using the second half of the proof of {Proposition~\ref{prop:buyer:equivalence-construction} with} $y_0:=z=-p^b_i(\xi,\mathcal{E})e^i$. Such a strategy $(\check{y},\check{\tau})$ superhedges $(-\xi_{\check{\tau}},\check{\tau})$ for the seller, so
\[
 p^b_i(\xi,\mathcal{E}) = -\check{y}^i_0 \le - p^a_i(-\xi_{\check{\tau}},\check{\tau}) \le p^b_i(\xi,\mathcal{E}),
\]
whence
\begin{equation} \label{eq:buyer-seller-eq}
-p^a_i(-\xi,\mathcal{E}^{\check{\tau}}) = p^b_i(\xi,\mathcal{E}).
\end{equation}
Thus the construction in the proof of Proposition~\ref{prop:construction-of-optimal-martingale-triple} of the optimal stopping time and approximate martingale pair for the seller of the European option $(-\xi_{\check{\tau}},\check{\tau})$ can be used to construct $\check{\chi}$ and $(\check{\mathbb{Q}},\check{S})\in\bar{\mathcal{P}}^i(\check{\chi})$ such that
\[
 \mathbb{E}_{\check{\mathbb{Q}}}((-\xi\cdot \check{S})_{\check{\chi}})=p^a_i(-\xi_{\check{\tau}},\check{\tau}).
\]
It is moreover clear from the construction in the proof of Proposition~\ref{prop:construction-of-optimal-martingale-triple} and the structure of the exercise policy of $(-\xi_{\check{\tau}},\check{\tau})$ that $\check{\chi}=\chi^{\check{\tau}}$. Thus $(\check{\mathbb{Q}},\check{S})\in\bar{\mathcal{P}}^i(\check{\tau})$ and
\[
  \mathbb{E}_{\check{\mathbb{Q}}}((\xi\cdot \check{S})_{\check{\tau}}) = \mathbb{E}_{\check{\mathbb{Q}}}((\xi\cdot \check{S})_{\check{\chi}}) =-\mathbb{E}_{\check{\mathbb{Q}}}((-\xi\cdot \check{S})_{\check{\chi}}) = -p^a_i(-\xi_{\check{\tau}},\check{\tau}) = p^b_i(\xi,\mathcal{E})
\]
as required.
\end{proof}

\begin{example}
\label{Exl:3}\upshape
Consider the computation of the bid price of the American option in Example \ref{Exl:2} using Construction \ref{const:buyer}. In contrast to the seller's case, some of the sets $\mathcal{U}_{t}^{b},\mathcal{V}_{t}^{b},\mathcal{W}_{t}^{b},\mathcal{Z}_{t}^{b}$ involved in this construction may fail to be convex, and there is no convex dual representation like that for the seller in Figure~\ref{algorithm_seller.png}.
To visualise the sets $\mathcal{U}_{t}^{b},\mathcal{V}_{t}^{b},\mathcal{W}_{t}^{b},\mathcal{Z}_{t}^{b}$ we just draw their boundaries.

The construction for the buyer proceeds as follows:
\begin{itemize}
\item The first step is to compute $\mathcal{Z}_{1}^{b}=\mathcal{U}_{1}^{b}$
in each of the four scenarios; see Figure~\ref{algorithm_buyer_u1s.png}.

\item Then we take the intersection of $\mathcal{U}_{1}^{b\omega_{1}}$,
$\mathcal{U}_{1}^{b\omega_{2}}$, $\mathcal{U}_{1}^{b\omega_{3}}$,
$\mathcal{U}_{1}^{b\omega_{4}}$ to obtain $\mathcal{W}_{0}^{b}$. This set
appears in Figure~\ref{algorithm_buyer.png}(a).

\item Next, the set $\mathcal{V}_{0}^{b}$ in Figure~\ref{algorithm_buyer.png}%
(b) is the sum $\mathcal{W}_{0}^{b}+\mathcal{K}_{0}$ of
$\mathcal{W}_{0}^{b}$ and the solvency cone~$\mathcal{K}_{0}$.

\item Then we take $\mathcal{U}_{0}^{b}$, which appears in
Figure~\ref{algorithm_buyer.png}(c).

\item Finally, the set $\mathcal{Z}_{0}^{b}$ is the union of $\mathcal{V}%
_{0}^{b}$ and $\mathcal{U}_{0}^{b}$. It appears in
Figure~\ref{algorithm_buyer.png}(d); the dark gray region belongs to
$\mathcal{V}_{0}^{b}$ (but not $\mathcal{U}_{0}^{b}$), and the light gray
region belongs to $\mathcal{U}_{0}^{b}$ (but not $\mathcal{V}_{0}^{b}$).%
\end{itemize}

\begin{figure}[th]
\includegraphics[
natheight=1.179600in,
natwidth=4.738300in,
height=1.2142in,
width=4.7919in
]%
{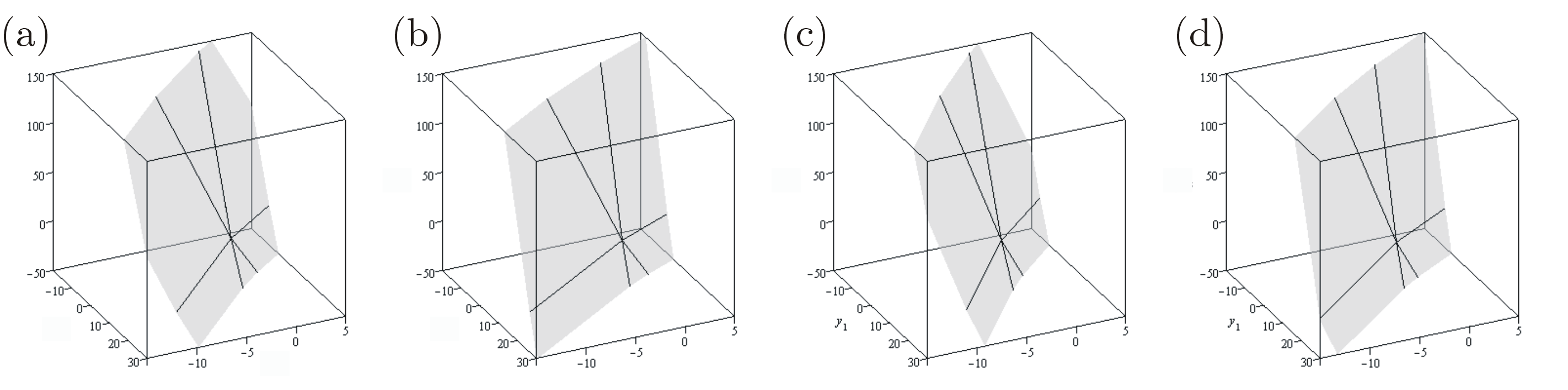}%
\caption{Sets $\mathcal{U}_{1}^{b\omega_{1}}$,
$\mathcal{U}_{1}^{b\omega_{2}}$, $\mathcal{U}_{1}^{b\omega_{3}}$,
$\mathcal{U}_{1}^{b\omega_{4}}$ in Construction \ref{const:buyer}, Example~\ref{Exl:3}}%
\label{algorithm_buyer_u1s.png}%
\end{figure}
\begin{figure}[th]
\includegraphics[
natheight=1.179600in,
natwidth=4.734800in,
height=1.2142in,
width=4.7885in
]%
{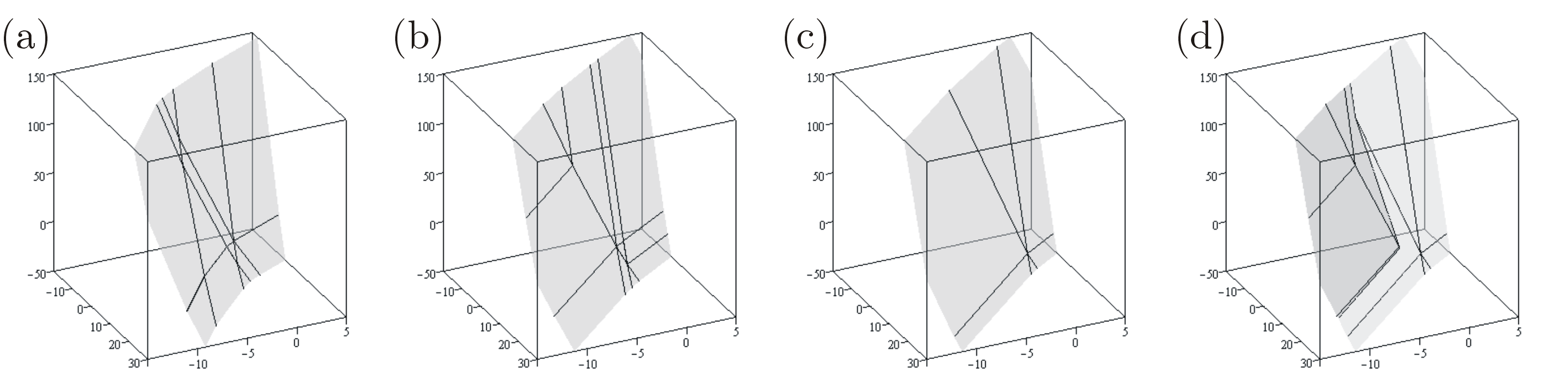}%
\caption{Sets $\mathcal{W}_{0}^{b}$, $\mathcal{V}_{0}^{b}$, $\mathcal{U}_{0}^{b}$
and $\mathcal{Z}_{0}^{b}$ in Construction \ref{const:buyer}, Example~\ref{Exl:3}}%
\label{algorithm_buyer.png}%
\end{figure}
\noindent The unbounded and non-convex set $\mathcal{Z}_{0}^{b}$ has $8$
vertices. Of these, the point $(-1,1,-33)$ is a vertex of $\mathcal{U}_{0}%
^{b}$, the points $(4,-13/2,163/2)$ and $(4,-15/7,-10)$ are vertices of~$\mathcal{V}_{0}^{b}$, and
\begin{gather*}
(-1,-39/7,361/3),(19/5,-15/7,-23/3),(39/10,-73/35,-10),\\
(4,-233/112,-89/8),(127/30,-15/7,-12)
\end{gather*}
are common to both $\mathcal{U}_{0}^{b}$ and~$\mathcal{V}_{0}^{b}$. The lowest
number $x$ such that $(0,0,x)\in\mathcal{Z}_{0}^{b}$ is $x=-\frac{59}{3}\in
\mathcal{U}_{0}^{b}$. By Theorem~\ref{th:buyer}, the bid price of the option is $-x=\frac{59}{3}\cong19.67$.

\end{example}

\section{Numerical examples\label{sec:numerical}}

{We now use the methods developed in this paper to study two examples with a realistic flavour in some detail.}

\begin{example} \label{ex:1}
 Consider a binomial tree model with two risky assets. We assume a notional friction-free exchange rate~$E=(E_{t})$ between the two
assets satisfying
\[
E_{t+1}=\varepsilon_{t}E_{t}%
\]
for~$t=0,\ldots,T-1$, where~$E_{0}=100$ is given, and where~{$(\varepsilon_{t})$}
is a sequence of independent identically distributed random variables taking
the values
\[
e^{\kappa{\Delta}+\sigma\sqrt{{\Delta}}},\quad
e^{\kappa{\Delta}-\sigma\sqrt{{\Delta}}},
\]
each with positive probability. Here~{$\Delta:=\frac{1}{T}$,} $\sigma=0.1$ is the volatility of the
exchange rate, $\kappa=0.05$ is the depreciation rate of the first asset in
terms of the second, the time horizon is $1$~year and~$T=250$ is the number of
steps in the model. We further assume that for~$t=0,\ldots,T$ the actual
exchange rates between the assets are
\[
\pi_{t}^{12}=(1+k)E_{t},\quad\pi_{t}^{11}=\pi_{t}^{22}=1,\quad\pi_{t}%
^{21}=\frac{1}{(1-k)E_{t}},
\]
where~$k=0.5\%$ is the transaction cost rate. A portfolio $y_{t}=(y_{t}%
^{1},y_{t}^{2})$ is solvent at time $t$ if and only if
\begin{equation} \label{eq:num:vartheta}
\vartheta_t(y_t):=\min\{y_{t}^{1}\pi_{t}^{21}+y_{t}^{2},y_{t}^{1}+y_{t}^{2}\pi_{t}^{12}\}\geq0.
\end{equation}

In friction-free models the owner of an option benefits from exercising it if and only if the option payoff can be converted into a non-negative number of units of one of the assets (and for this reason it is standard practice to represent options in friction-free models as non-negative cash payoffs). In the presence of transaction costs, where assets are not freely exchangeable, the situation is no longer so clear-cut, since the benefit from receiving a payoff consisting of a portfolio of assets {depends} greatly on the current position held in the underlying assets at the time that the payoff becomes available. Motivated by the work of \cite{perrakis_lefoll2004}, we make no assumption on the form of the payoff itself but award the owner of an option the right to not exercise the option at all. This is done by formally adding an extra time step $T+1$ in the model and setting the option payoff at that time to be zero.

Consider an American call option on the second asset with expiration date $T$, strike $100$ and physical delivery. This corresponds to the payoff process $\xi=(\xi_t)$ with
\[
\xi_{t}=(\xi_{t}^{1},\xi_{t}^{2})=(-100,1)
\]
for~$t\le T$ and $\xi_{T+1}=(0,0)$. We say that the option is \emph{in the money} {at time~$t$} if $\xi_t$ is a solvent portfolio at time $t$, and \emph{out of the money} if it isn't. {An implementation in C++ of} Constructions~\ref{const:seller} and~\ref{const:buyer} (see also Section~\ref{sec:American}) {gives the ask and bid prices as}
\begin{align*}
 p^a_1(\xi) &=  6.67776, & p^b_1(\xi) &= 0.101895.
\end{align*}

It is interesting to note that the optimal stopping times for the buyer and seller of the American call option are by no means unique, and also that the sets of optimal stopping times for the buyer and seller differ. To see this, consider the two scenarios $\alpha$ and $\beta$ depicted in Figure~\ref{fig:1}. The asset price histories associated with $\alpha$ and $\beta$ coincide up to time step $191$. In scenario $\alpha$, the option is in the money at all times $t$ after step $153$, whereas in scenario $\beta$ the option moves out of the money at time step $225$ and stays out of the money until maturity.

\begin{figure}
\begin{tikzpicture}[x=\tikzwidth/275*\textwidth,y=\tikzheight/60*\textwidth]
	\draw[axis,->] (0,80) -- (275,80) node[right] {$t$};
	\draw[axis,dashed] (0,100) -- (275,100);
	\foreach \x in {50, 100, 154, 191, 225, 250}
		\draw (\x,80) node[below] {$\x$} -- ++(0,1ex);
	\draw[axis,->] (0,80) -- (0,140);
	\foreach \y in {80, 90, 100, 110, 120, 130}
		\draw (0,\y) node[left] {$\y$} -- ++(1ex,0);
	\draw[function] plot file {bin1-buyer-ask} node[above] {$\pi^{12}_t(\alpha)$};
	\draw[function] plot file {bin1-buyer-bid} node[right] {$1/\pi^{21}_t(\alpha)$};
	\draw[function] plot file {bin2-buyer-ask} node[right] {$\pi^{12}_t(\beta)$};
	\draw[function] plot file {bin2-buyer-bid} node[below] {$1/\pi^{21}_t(\beta)$};
\end{tikzpicture}
\caption{Asset prices in binary model, Example~\ref{ex:1}}%
\label{fig:1}%
\end{figure}
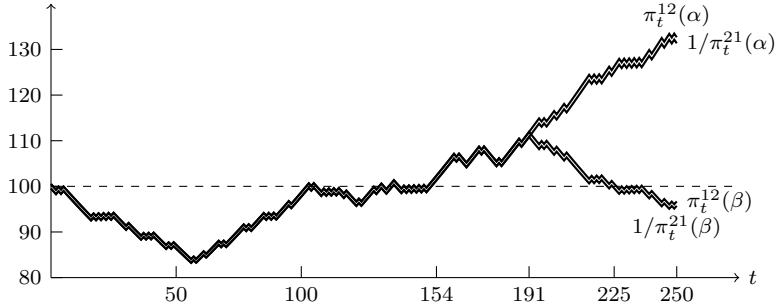

Consider first the superhedging problems for the seller of the American call option in these two {scenarios}. The optimal superhedging strategy $y=(y^1_t,y^2_t)$ for the seller can be constructed as in the proof of Proposition \ref{prop:seller:equivalence-construction} from an initial endowment of $(6.67776,0)$; see Figure~\ref{fig:2}. The optimal stopping time $\chi$ for the seller can be constructed as in the proof of Proposition \ref{prop:construction-of-optimal-martingale-triple}; see Figure~\ref{fig:3}.

In scenario $\alpha$, where the option matures in the money, the optimal superhedging strategy for the seller converges to the option payoff; in particular, it becomes a static strategy $(y^1_t(\alpha),y^2_t(\alpha))=(-100,1)$ for $t\ge218$ in this scenario. This coincides with the earliest time instant when the optimal stoping time $\chi_t(\alpha)$ becomes non-zero ($\chi^*_t(\alpha)$ becomes less than~$1$). Figure~\ref{fig:3} depicts one possibility, but note that the optimal stopping time for the seller is highly non-unique on this path.

In scenario $\beta$, where the option matures out of the money, the optimal superhedging strategy for the seller converges to zero; in particular $(y^1_t,y^2_t)=(0,0)$ for $t\ge247$. This feature results from the need for the seller to remain solvent in the event that the buyer never exercises the {option, which} is likely if the option is both close to maturity and out of the money. The amount of trading required to transform the asset holdings in scenario $\beta$ from a superhedging to a solvent position over the latter part of the model attracts high transaction costs, with the result that the optimal stopping time for the seller, shown in Figure~\ref{fig:3}, corresponds to the buyer never exercising the option.

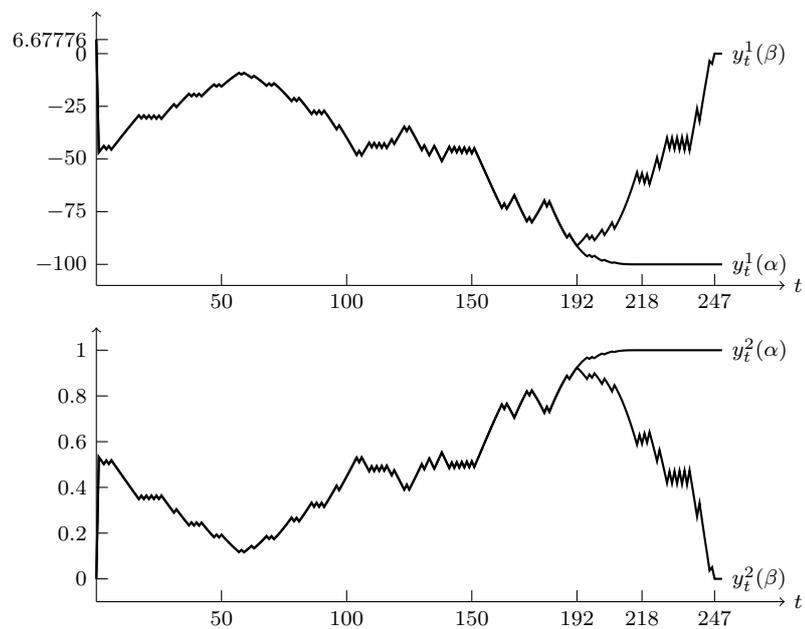
\begin{figure}
\begin{tabular}{@{}r@{}}
\begin{tikzpicture}[x=\tikzwidth/275*\textwidth,y=\tikzheight/130*\textwidth]
	\draw[axis,->] (0,-110) -- (275,-110) node[right] {$t$};
	\foreach \x in {50, 100, 150, 192, 218,247}
		\draw (\x,-110) node[below] {$\x$} -- ++(0,1ex);
	\draw[axis,->] (0,-110) -- (0,20);
	\foreach \y in {-100, -75, -50, -25, 0, 6.67776}
		\draw (0,\y) node[left] {$\y$} -- ++(1ex,0);
	\draw[function] plot file {bin1-seller-cash} node[right] {$y^1_t(\alpha)$};
	\draw[function] plot file {bin2-seller-cash} node[right] {$y^1_t(\beta)$};
\end{tikzpicture}\\
\begin{tikzpicture}[x=\tikzwidth/275*\textwidth,y=\tikzheight/1.2*\textwidth]
	\draw[axis,->] (0,-0.1) -- (275,-0.1) node[right] {$t$};
	\foreach \x in {50, 100, 150, 192, 218,247}
		\draw (\x,-0.1) node[below] {$\x$} -- ++(0,1ex);
	\draw[axis,->] (0,-0.1) -- (0,1.1);
	\foreach \y in {0,0.2,0.4,0.6,0.8,1}
		\draw (0,\y) node[left] {$\y$} -- ++(1ex,0);
	\draw[function] plot file {bin1-seller-shares} node[right] {$y^2_t(\alpha)$};
	\draw[function] plot file {bin2-seller-shares} node[right] {$y^2_t(\beta)$};
\end{tikzpicture}
\end{tabular}
\caption{Optimal superhedging strategy for seller of American call option, Example~\ref{ex:1}}%
\label{fig:2}%
\end{figure}

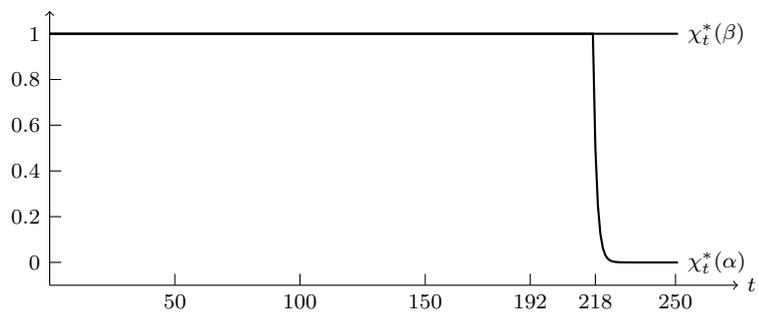
\begin{figure}
\begin{tikzpicture}[x=\tikzwidth/275*\textwidth,y=\tikzheight/1.2*\textwidth]
	\draw[axis,->] (0,-0.1) -- (275,-0.1) node[right] {$t$};
	\foreach \x in {50, 100, 150, 192, 218,250}
		\draw (\x,-0.1) node[below] {$\x$} -- ++(0,1ex);
	\draw[axis,->] (0,-0.1) -- (0,1.1);
	\foreach \y/\yext in {0,0.2,0.4,0.6,0.8,1/\phantom{6.67776}1}
		\draw (0,\y) node[left] {$\yext$} -- ++(1ex,0);
	\draw[function] plot file {bin1-seller-chiast} node[right] {$\chi^\ast_t(\alpha)$};
	\draw[function] plot file {bin2-seller-chiast} node[right] {$\chi^\ast_t(\beta)$};
\end{tikzpicture}
\caption{Optimal stopping time for seller of American call option, Example~\ref{ex:1}}%
\label{fig:3}%
\end{figure}

Consider now the superhedging and optimal exercise problems for the buyer of the American call. The optimal superhedging strategy $y=(y_t)$ and optimal stopping time $\tau$ can be constructed as in the last part of the proof of Proposition~\ref{prop:buyer:equivalence-construction}. The values of the optimal superhedging strategy in {scenarios} $\alpha$ and~$\beta$ are depicted in Figure~\ref{fig:4}.

The construction in the proof of Proposition \ref{prop:buyer:equivalence-construction} gives the optimal exercise time for the buyer in these {scenarios} as
\[\tau(\alpha)=\tau(\beta)=137.\]
At first glance this appears to be contrary to the received wisdom that it is never optimal to exercise an American call early. There is however no contradiction; it is rather the case that the optimal exercise time is not unique and this particular construction returns the earliest optimal stopping time. In particular, recall that the optimal stopping time $\tau$ constructed in Proposition~\ref{prop:buyer:equivalence-construction} is the first stopping time at which the buyer can exercise the option and remain solvent, i.e.
\[
 \tau = \min\{t:\vartheta_t(y_t+\xi_t)\ge0\},
\]
where $\vartheta_t$ is given by \eqref{eq:num:vartheta}. The values of $\vartheta_t$ in {scenarios} $\alpha$ and $\beta$ appear in Figure~\ref{fig:5}, {which} confirms why the first optimal exercise time in these {scenarios} should {be}~$137$.

\begin{figure}
\begin{tabular}{@{}r@{}}
\begin{tikzpicture}[x=\tikzwidth/275*\textwidth,y=\tikzheight/30*\textwidth]
	\draw[axis,->] (0,-5) -- (275,-5) node[right] {$t$};
	\foreach \x in {50, 100, 137,154, 192,225,250}
		\draw (\x,-5) node[below] {$\x$} -- ++(0,1ex);
	\draw[axis,->] (0,-5) -- (0,25);
	\foreach \y in {-0.101895, 5, 10, 15, 20}
		\draw (0,\y) node[left] {$\y$} -- ++(1ex,0);
	\draw[function] plot file {bin1-buyer-cash} node[right] {$y^1_t(\alpha)$};
	\draw[function] plot file {bin2-buyer-cash} node[right] {$y^1_t(\beta)$};
\end{tikzpicture}\\
\begin{tikzpicture}[x=\tikzwidth/275*\textwidth,y=\tikzheight/0.25*\textwidth]
	\draw[axis,->] (0,-0.225) -- (275,-0.225) node[right] {$t$};
	\foreach \x in {50, 100, 137,154, 192,225,250}
		\draw (\x,-0.225) node[below] {$\x$} -- ++(0,1ex);
	\draw[axis,->] (0,-0.225) -- (0,0.025);
	\foreach \y in {-0.2,-0.15,-0.1,-0.05,0}
		\draw (0,\y) node[left] {$\y$} -- ++(1ex,0);
	\draw[function] plot file {bin1-buyer-shares} node[right] {$y^2_t(\alpha)$};
	\draw[function] plot file {bin2-buyer-shares} node[right] {$y^2_t(\beta)$};
\end{tikzpicture}
\end{tabular}
\caption{Optimal superhedging strategy for buyer of American call, Example~\ref{ex:1}}%
\label{fig:4}%
\end{figure}
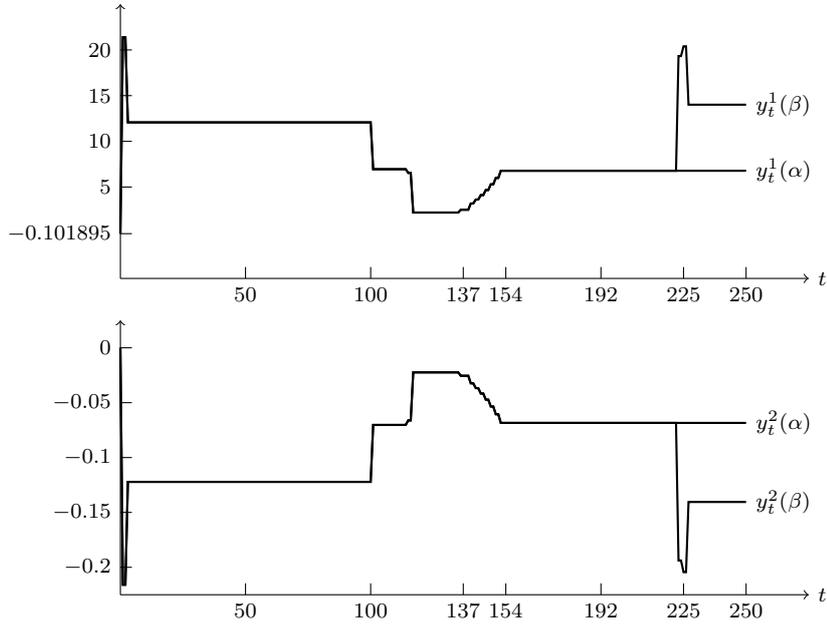

\begin{figure}
\begin{tikzpicture}[x=\tikzwidth/275*\textwidth,y=\tikzheight/55*\textwidth]
	\draw[axis,dashed] (0,0) -- (275,0);
	\draw[axis,->] (0,-20) -- (275,-20) node[right] {$t$};
	\foreach \x in {50, 100, 137, 154, 191, 225,250}
		\draw (\x,-20) node[below] {$\x$} -- ++(0,1ex);
	\draw[axis,->] (0,-20) -- (0,35);
	\foreach \y/\yext in {-10,0,10,20,30/\phantom{-0.101895},30}
		\draw (0,\y) node[left] {$\yext$} -- ++(1ex,0);
	\draw[function] plot file {bin1-buyer-theta} node[above] {$\vartheta_t(y_t+\xi_t)(\alpha)$};
	\draw[function] plot file {bin2-buyer-theta} node[below] {$\vartheta_t(y_t+\xi_t)(\beta)$};
\end{tikzpicture}
\caption{Exercise attractiveness for buyer of American call, Example~\ref{ex:1}}%
\label{fig:5}%
\end{figure}
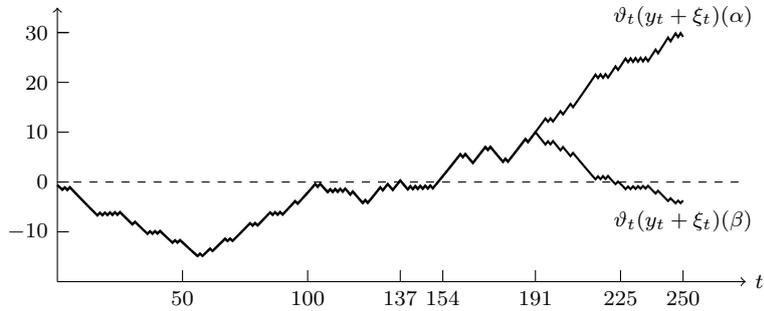
\end{example}

 \begin{example}\label{ex:2}
Consider a model with three currencies and $T=4$ steps with time horizon $1$ based on the two-asset recombinant Korn-Muller model \cite{Korn_Muller2009} with Cholesky decomposition, that is, consider the process $(S_t)$ with
\[
 S_{t+1}=(\varepsilon^1_t S^1_t, \varepsilon^2_t S^2_t,1)
\]
for $t<T$, where $\Delta=\frac{1}{T}$ is the step size and $(S^1_0,S^2_0)=(40,=50)$, and where $(\varepsilon_t)=(\varepsilon^1_t,\varepsilon^2_t)$ is a sequence of independent identically distributed random variables taking the values
\begin{align*}
 \left(e^{-\frac{1}{2}\sigma_1^2\Delta-\sigma_1\sqrt\Delta},e^{-\frac{1}{2}\sigma_2^2\Delta-(\rho+\sqrt{1-\rho^2})\sigma_2\sqrt\Delta}\right),\\
  \left(e^{-\frac{1}{2}\sigma_1^2\Delta-\sigma_1\sqrt\Delta},e^{-\frac{1}{2}\sigma_2^2\Delta-(\rho-\sqrt{1-\rho^2})\sigma_2\sqrt\Delta}\right),\\
 \left(e^{-\frac{1}{2}\sigma_1^2\Delta+\sigma_1\sqrt\Delta},e^{-\frac{1}{2}\sigma_2^2\Delta+(\rho-\sqrt{1-\rho^2})\sigma_2\sqrt\Delta}\right),\\
  \left(e^{-\frac{1}{2}\sigma_1^2\Delta+\sigma_1\sqrt\Delta},e^{-\frac{1}{2}\sigma_2^2\Delta+(\rho+\sqrt{1-\rho^2})\sigma_2\sqrt\Delta}\right),
\end{align*}
each with positive probability. Here $\sigma_1=0.15$, $\sigma_2=0.1$ and $\rho=0.5$. The exchange rates with transaction costs are modelled as
\[
 \pi^{ij}_t:=
 \begin{cases}
 \frac{S^j_t}{S^i_t}(1+k) &\text{if }i\neq j,\\
 1 & \text{if }i=j,
 \end{cases}
\]
for $i,j=1,\ldots,3$ and $t\le T$, where $k=0.005$.

The pricing and hedging constructions of Section \ref{sec:main-results} was implemented by means of the Maple \emph{Convex} package {\cite{Franz2009}} for an American put option with physical delivery on a basket containing one unit each of the first two currencies and with strike $95$ in the third currency, i.e.
\[
 \xi_t=(\xi_1^1,\xi_t^2,\xi_t^3)=(-1,-1,95)
\]
for $t\le T$. As in the previous example we allow for the possibility that the option holder may refrain from exercising by adding an additional time step $T+1$ and taking $\xi_{T+1}=(0,0,0)$. Constructions \ref{const:seller} and \ref{const:buyer} give the ask and bid prices of this option in the three currencies as
\begin{align*}
 p^a_1(\xi) &= 0.22587, & p^a_2(\xi) &= 0.18070, & p^a_3(\xi) &= 8.98997,\\
 p^b_1(\xi) &= 0.12075, & p^b_2(\xi) &= 0.09660, & p^b_3(\xi) &= 4.85420.
\end{align*}

Let us now use Constructions \ref{const:seller:hedge} and \ref{const:buyer:hedge} to compute the hedging strategies for the buyer and seller in the scenario corresponding to the path
\begin{align*}
 S_0 = \begin{pmatrix}40\\50\end{pmatrix},
 S_1 = \left(\begin{array}{d}37.006\\46.641\end{array}\right),
 S_2 = \left(\begin{array}{d}34.235\\47.443\end{array}\right),
 S_3 = \left(\begin{array}{d}36.798\\50.733\end{array}\right),
 S_4 = \left(\begin{array}{d}39.553\\54.251\end{array}\right).
\end{align*}
Table \ref{tab:num2:1} gives the resulting strategy for the seller starting from the initial endowment $p^a_3(\xi)e^3$, with the bullet in each graph representing $y_t$. For $t=0,1$ the set $(y_t-\mathcal{K}_t)\cap\mathcal{W}^a_t$ has only one element, which becomes $y_{t+1}$. For $t=2,3$ we have $y_t\in(y_t-\mathcal{K}_t)\cap\mathcal{W}^a_t$ and so it was natural to let $y_{t+1}:=y_t$ to avoid trading (and the associated transaction costs). For $t=2,3$ the choice of $y_{t+1}$ is no longer unique; the choice $y_{t+1}:=y_t$ in Table \ref{tab:num2:1} avoids trading (and the associated transaction costs) but any other element of $(y_t-\mathcal{K}_t)\cap\mathcal{W}^a_t$ would have been acceptable in each case.
\begin{table}
\caption{Superhedging {strategy} for seller along a path, Example~\ref{ex:2}}
\label{tab:num2:1}
\begin{tabular}{@{}cc@{}c@{}c@{}c@{}}
\hline\noalign{\smallskip}
 $t$ & $S_t$ & $y_t$ & $\mathcal{Z}^a_t$ & $(y_t-\mathcal{K}_t)\cap\mathcal{W}^a_t$\\
 \noalign{\smallskip}\hline\noalign{\smallskip}
 0 & $\left(\begin{array}{d}40.000\\50.000\end{array}\right)$ & $\left(\begin{array}{d}0.000\\0.000\\8.990\end{array}\right)$ & \raisebox{-0.5\height}{\includegraphics[height=\tabfigheight]{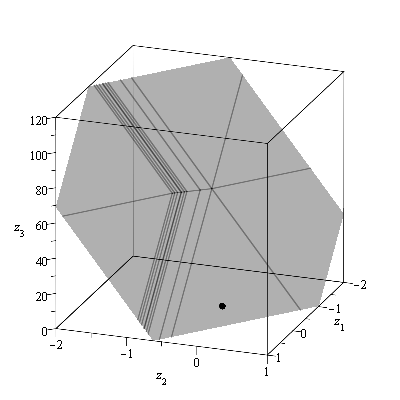}} & $\left\{\left(\begin{array}{d}-0.798\\-0.440\\62.668\end{array}\right)\right\}$ \\
 1 & $\left(\begin{array}{d}37.006\\46.641\end{array}\right)$ & $\left(\begin{array}{d}-0.798\\-0.440\\62.668\end{array}\right)$ & \raisebox{-0.5\height}{\includegraphics[height=\tabfigheight]{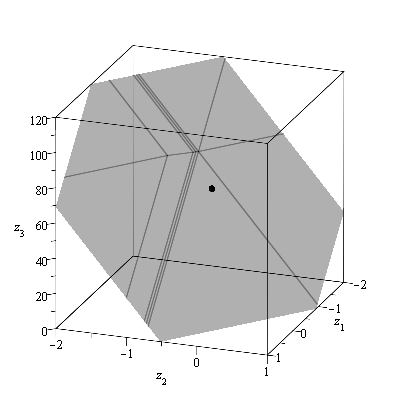}} & $\left\{\left(\begin{array}{d}-0.962\\-0.690\\80.300\end{array}\right)\right\}$\\
 2 & $\left(\begin{array}{d}34.235\\47.443\end{array}\right)$ & $\left(\begin{array}{d}-0.962\\-0.690\\80.300\end{array}\right)$ & \raisebox{-0.5\height}{\includegraphics[height=\tabfigheight]{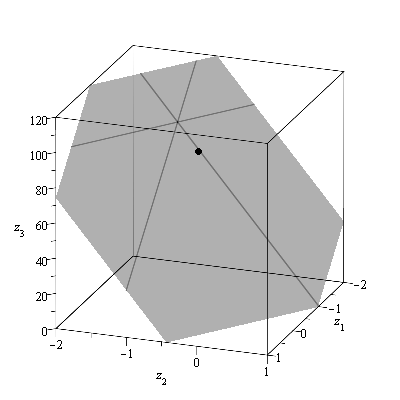}} & \raisebox{-0.5\height}{\includegraphics[height=\tabfigheight]{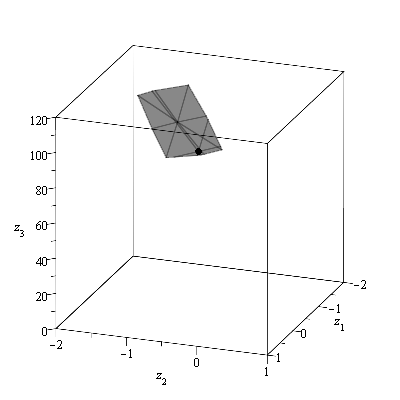}} \\
 3 & $\left(\begin{array}{d}36.798\\50.733\end{array}\right)$ & $\left(\begin{array}{d}-0.962\\-0.690\\80.300\end{array}\right)$ & \raisebox{-0.5\height}{\includegraphics[height=\tabfigheight]{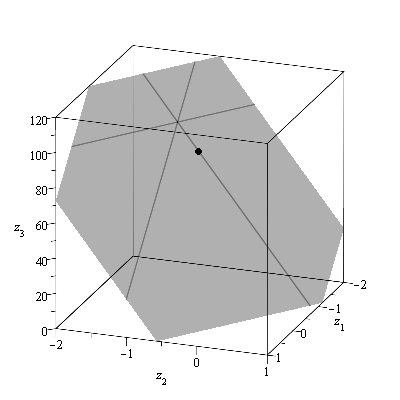}} & \raisebox{-0.5\height}{\includegraphics[height=\tabfigheight]{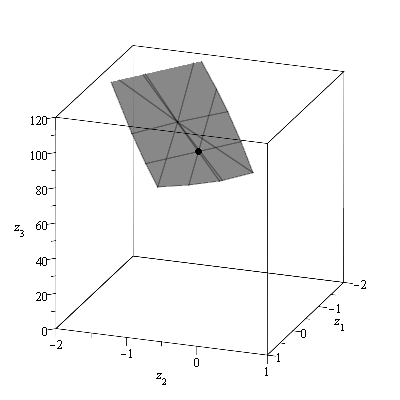}} \\
 4 & $\left(\begin{array}{d}39.553\\54.251\end{array}\right)$ & $\left(\begin{array}{d}-0.962\\-0.690\\80.300\end{array}\right)$ & \raisebox{-0.5\height}{\includegraphics[height=\tabfigheight]{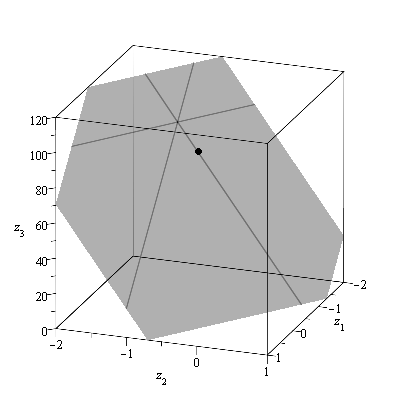}} & N/A \\
\noalign{\smallskip}\hline
 \end{tabular}
\end{table}

Table \ref{tab:num2:2} gives the optimal strategy for the buyer starting from the initial endowment $-p^b_3(\xi)e^3$ along the same path (omitted to save space). Again the bullet in each graph represents $y_t$. Since $y_0=-p^b_3(\xi)e^3\notin\mathcal{U}^b_0$ we have $\tau_0=1$, which reflects that it is in the buyer's best interest to wait rather than exercise the option at time $0$. Since $y_1\in\mathcal{U}^b_1=\mathcal{Z}^b_1$ we have $\tau_1=1$, which means that in this path it is optimal to exercise the option at time $1$. Construction \ref{const:buyer:hedge} completes the strategy by formally setting $y_4=y_3=y_2$ and $\tau_4=\tau_3=\tau_2$, but in practice a market agent exercising the option at time $1$ would create the portfolio
\[
 y_1 + \xi_1 = (-1,-1,95) = (-0.201,-0.365,26.143)\in\mathcal{K}_1
\]
and liquidate it immediately (for example, into $1.535$ units of currency $3$).
\begin{table}
\caption{Superhedging {strategy} for buyer along a path, Example~\ref{ex:2}}
\label{tab:num2:2}
\begin{tabular}{@{}cc@{}c@{}ccc@{}}
\hline\noalign{\smallskip}
 $t$ & $y_t$ & $\mathcal{Z}^b_t$ & $y_t\in\mathcal{U}^b_t$? & $\tau_t$ & $(y_t-\mathcal{K}_t)\cap\mathcal{W}^b_t$\\
 \noalign{\smallskip}\hline\noalign{\smallskip}
 $\left.0\right.$ & $\left(\begin{array}{d}0.000\\0.000\\-4.854\end{array}\right)$ & \raisebox{-0.5\height}{\includegraphics[height=\tabfigheight]{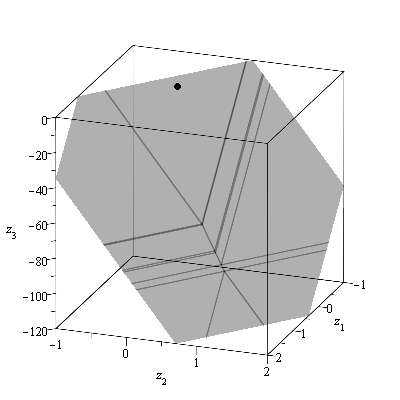}} & No & $1$ & $\left\{\left(\begin{array}{d}0.799\\0.635\\-68.857\end{array}\right)\right\}$ \\
 $\left.1\right.$ & $\left(\begin{array}{d}0.799\\0.635\\-68.857\end{array}\right)$ & \raisebox{-0.5\height}{\includegraphics[height=\tabfigheight]{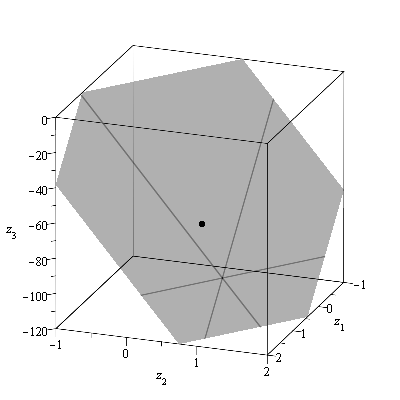}} & Yes & 1 & N/A \\
 2--4 & $\left(\begin{array}{d}0.799\\0.635\\-68.857\end{array}\right)$ & N/A & N/A & 1 & N/A \\
\noalign{\smallskip}\hline
\end{tabular}
\end{table}

 \end{example}

\appendix

\section{Appendix: Proof of Lemma \protect\ref{lem:support-technical}}
\label{appendix}

\normalsize

Lemma \ref{lem:support-technical} in Section \ref{sec:seller} depends on the following technical result.

\begin{lemma} \label{lem:convex-hull-support-functions}
Fix some $i=1,\ldots,d$, and let $A_1,\ldots A_n$ be non-empty closed convex sets in $\mathbb{R}^d$ such that $\dom\delta^\ast_{A_k}$ is compactly $i$-generated for all $k$. Define $A:=\bigcap_{k=1}^n A_k\neq\emptyset$; then
\[
 \delta^\ast_A = \conv\{\delta^\ast_{A_1},\ldots,\delta^\ast_{A_n}\},
\]
and for each $x\in\sigma_i(\dom\delta^\ast_{A})$ there exist $\alpha_1,\ldots,\alpha_n\ge0$ and $x_1,\ldots,x_n$ with $x_k\in\sigma_i(\dom\delta^\ast_{A_k})$ for all $k$ such that
\begin{align*}
 \delta^\ast_A(x) &= \sum_{k=1}^n\alpha_k\delta^\ast_{A_k}(x_k), & \sum_{k=1}^n\alpha_k &= 1, & \sum_{k=1}^n\alpha_kx_k &= x.
\end{align*}
The cone $\dom\delta^\ast_{A}$ is moreover compactly $i$-generated and
\begin{equation} \label{eq:lem:convex-hull-support-functions:4}
 \dom\delta^\ast_{A} = \conv\left[\bigcup_{k=1}^n\dom\delta^\ast_{A_k}\right].
\end{equation}
\end{lemma}

\begin{proof}
Let $f:=\conv\{\delta^\ast_{A_1}, \ldots, \delta^\ast_{A_n}\}$. Then
$
 \cl f = \delta^\ast_A;
$
see {\cite[Corollary 16.5.1]{rockafellar1996}}. Since $\delta^\ast_A$ is proper it follows that $f$ is proper and
\begin{equation} \label{eq:lem:convex-hull-support-functions:epi f}
 \overline{\epi f} = \epi \delta^\ast_A
\end{equation}
by \eqref{eq:closure-of-function}, so that $\delta^\ast_A=\cl f\le f$.

For any $k=1,\ldots,n$, the compact $i$-generation of $\dom \delta^\ast_{A_k}$ means that $\sigma_i(\dom\delta^\ast_{A_k})$ is compact and non-empty. Thus the positive homogeneity of $\delta^\ast_{A_k}$ guarantees the existence of a closed proper convex function $g_k$ with $\dom g_k=\sigma_i(\dom\delta^\ast_{A_k})$ compact such that $\delta^\ast_{A_k}$ is \emph{generated} by $g_k$, i.e.
\[
 \delta^\ast_{A_k}(y) = \begin{cases}
 \lambda g_k(x) &\text{if there exists } \lambda\ge0\text{ and }x\in\dom g_k\text{ such that }y=\lambda x,\\
\infty & \text{otherwise}.
\end{cases}
\]

Let $g:=\conv\{g_1,\ldots,g_n\}$; then
\[
 \dom g = \conv\left[\bigcup_{k=1}^n\sigma_i(\dom\delta^\ast_{A_k})\right]
\]
is compact {\cite[Corrolary 9.8.2]{rockafellar1996}}. Moreover, $g$ is closed and proper, and for each $x\in\dom g$ there exist $\alpha_1,\ldots,\alpha_n\ge0$ and $x_1,\ldots,x_n$ such that $x_k\in\sigma_i(\dom\delta^\ast_{A_k})$ for all $k$ and
\begin{align}
 g(x) &= \sum_{k=1}^n\alpha_kg_k(x_k), & \sum_{k=1}^n\alpha_k &= 1, & \sum_{k=1}^n\alpha_kx_k &= x; \label{eq:lem:support-duality:1}
\end{align}
see \cite[Corollary 9.8.3]{rockafellar1996} (the common recession function is $\delta^\ast_{\mathbb{R}^d}$ since $\dom g_k$ is compact for all $k$).

Let $h$ be the positively homogeneous function generated by $g$, i.e.
\[
h(y)
:= \begin{cases}
 \lambda g(x) &\text{if there exists } \lambda\ge0\text{ and }x\in\dom g\text{ such that }y=\lambda x,\\
\infty & \text{otherwise.}
\end{cases}
\]
Clearly, $h$ is a proper convex function and $\dom h = \cone (\dom g)$ is compactly $i$-generated. The function $h$ is moreover closed since
\[
 \epi h = (\cone (\epi g)) \cup\{(0,\lambda):\lambda\ge0\} = \overline{\epi h};
\]
see {\cite[Theorem 9.6]{rockafellar1996}}, and it is majorised by $\delta^\ast_{A_1}, \ldots, \delta^\ast_{A_n}$, hence $h\le f$. Since~$h$ is closed, it then follows from \eqref{eq:lem:convex-hull-support-functions:epi f} that
\begin{equation} \label{eq:lem:convex-hull-support-functions:fh}
 h  \le \delta^\ast_A \le f.
\end{equation}

Fix any $y\in\dom h$. There exist $\lambda\ge0$ and $x\in\sigma_i(\dom h)=\dom g$ such that $y=\lambda x$. Fix any $\alpha_1,\ldots,\alpha_n\ge0$ and $x_1,\ldots,x_n$ satisfying \eqref{eq:lem:support-duality:1} and where $x_k\in\sigma_i(\dom\delta^\ast_{A_k})$ for all $k$. Let $y_k:=\lambda x_k$ for all $k$. Then
\[
 \sum_{k=1}^n \alpha_k y_k = \lambda\sum_{k=1}^n \alpha_k x_k = \lambda x = y
\]
and
\[
 \sum_{k=1}^n\alpha_k\delta^\ast_{A_k}(y_k) = \lambda\sum_{k=1}^n\alpha_kg_k(x_k) = \lambda g(x) = h(y).
\]
By the definition of the convex hull, this means that $f(y)\le h(y)$. Combining this with \eqref{eq:lem:convex-hull-support-functions:fh} gives
\begin{equation*}
 f = h = \delta^\ast_A.
\end{equation*}
The properties of $\dom\delta^\ast_{A}$, in particular~\eqref{eq:lem:convex-hull-support-functions:4}, then follow upon observing that \[\dom g=\sigma_i(\dom h)=\sigma_i(\dom\delta^\ast_{A}).
\]
\end{proof}

The paper concludes with the proof of Lemma \ref{lem:support-technical}.

\begin{proof}[Proof of Lemma \ref{lem:support-technical}]
For each $t$, since $\mathcal{K}_t$ is a cone, the support function of $-\mathcal{K}_t$ is
\begin{align}
 \delta^\ast_{-\mathcal{K}_t}(x)
&= \begin{cases}
    0 &\text{if } x\cdot y \le 0 \text{ for all } y\in-\mathcal{K}_t,\\
    \infty &\text{otherwise}
   \end{cases}
= \begin{cases}
    0 &\text{if } x\in\mathcal{K}^\ast_t,\\
    \infty &\text{otherwise}.\label{eq:support-of-St}
   \end{cases}
\end{align}
Thus $\dom \delta^\ast_{-\mathcal{K}_t}=\mathcal{K}_t^\ast$, and so $\dom \delta^\ast_{-\mathcal{K}_t}$ is compactly $i$-generated.

For any $t$ we have $U^a_t=\delta^\ast_{\mathbb{R}^d}$ on $\Omega\setminus\mathcal{E}_t$, together with
\[
 U^a_t(y) = \delta^\ast_{\{-\xi_t\}-\mathcal{K}_t}(y) = \delta^\ast_{\{-\xi_t\}}(y) + \delta^\ast_{-\mathcal{K}_t}(y) = -y\cdot\xi_t + \delta^\ast_{-\mathcal{K}_t}(y)
\]
for $y\in\mathbb{R}^d$ on $\mathcal{E}_t$ {\cite[p.~113]{rockafellar1996}}. Similarly,
\[
 V^a_t = \delta^\ast_{-\mathcal{W}^a_t-\mathcal{K}_t} = \delta^\ast_{-\mathcal{W}^a_t} + \delta^\ast_{-\mathcal{K}_t} = W^a_t + \delta^\ast_{-\mathcal{K}_t}.
\]
Equalities \eqref{eq:U^a_t-support-function-formula} and \eqref{eq:V^a_t-support-function-formula} then follow from \eqref{eq:lem:support-of-cone:Rd} and \eqref{eq:support-of-St}.

We now turn to claims~\ref{claim:lem:support-technical:2} and~\ref{claim:lem:support-technical:3}. Note first that the sets $\mathcal{U}^a_t$, $\mathcal{V}^a_t$, $\mathcal{W}^a_t$ and $\mathcal{Z}^a_t$ are non-empty for all $t$. This is easy to check by taking the trivial superhedging strategy for the seller defined by \eqref{eq:seller:trivial-superhedging-strategy} and following the backward induction argument in the proof of Proposition~\ref{prop:seller:equivalence-construction}.

We show below by backward induction that $\dom Z^a_t$ is compactly $i$-generated
on $\mathcal{E}^\ast_t.$ While doing so we will establish claims~\ref{claim:lem:support-technical:2} and~\ref{claim:lem:support-technical:3} for all $t$. At time~$T$, using $Z^a_T=U^a_T$ and \eqref{eq:construction:Ut}, the set $ \dom Z^a_T = \mathcal{K}_T^\ast$ is compactly $i$-generated on $\mathcal{E}^\ast_T=\mathcal{E}_T$, while $Z^a_T=\delta^\ast_{\mathbb{R}^d}$ on $\Omega\setminus\mathcal{E}^\ast_T$. This establishes claim~\ref{claim:lem:support-technical:2} for $t=T$ since $\mathcal{E}^\ast_{T+1}=\emptyset$.

At any time $t<T$, suppose that $\dom Z^a_{t+1}$ is compactly $i$-generated
on $\mathcal{E}^\ast_{t+1}$. For any $\mu\in\Omega_t$ there are now two possibilities:
\begin{itemize}
 \item If $\mu\subseteq\mathcal{E}^\ast_{t+1}$, then Lemma \ref{lem:convex-hull-support-functions} applies to the sets $\{-\mathcal{Z}^{a\nu}_{t+1}:\nu\in\successors\mu\}$ since
\[
 \bigcap_{\nu\in\successors\mu}\mathcal{Z}^{a\nu}_{t+1} = \mathcal{W}^{a\mu}_t \neq \emptyset;
\]
this immediately gives claim~\ref{claim:lem:support-technical:3}. Moreover, the compact $i$-generation of $\dom W^{a\mu}_t$ in combination with
\[
 \dom V^{a\mu}_t = \dom W^{a\mu}_t \cap \mathcal{K}^{\ast\mu}_t
\]
shows that $\dom V^{a\mu}_t$ is also compactly $i$-generated. There are now two possibilities:
\begin{itemize}
\item If $\mu\subseteq\mathcal{E}_t$, then Lemma \ref{lem:convex-hull-support-functions} applies to the sets $-\mathcal{U}^{a\mu}_t$ and $-\mathcal{V}^{a\mu}_t$.
This gives claim~\ref{claim:lem:support-technical:2}\ref{claim:lem:support-technical:2a} after noting that
\[
 \dom Z^{a\mu}_t = \conv(\dom V^{a\mu}_t\cup\mathcal{K}^{\ast\mu}_t)=\mathcal{K}^{\ast\mu}_t
\]
by \eqref{eq:lem:convex-hull-support-functions:4}.
\item If $\mu\not\subseteq\mathcal{E}_t$, then $Z^{a\mu}_t=V^{a\mu}_t$ by Remark \ref{remark:construction:when-is-R^d:dual}, which gives claim~\ref{claim:lem:support-technical:2}\ref{claim:lem:support-technical:2c}.
\end{itemize}

\item If $\mu\not\subseteq\mathcal{E}^\ast_{t+1}$, then $Z^{a\mu}_t=U^{a\mu}_t$ by Remark \ref{remark:construction:when-is-R^d:dual}. There are again two possibilities:
\begin{itemize}
\item If $\mu\subseteq\mathcal{E}_t$, then \eqref{eq:U^a_t-support-function-formula} gives $\dom Z^{a\mu}_t=\mathcal{K}^{\ast\mu}_t$. This is claim~\ref{claim:lem:support-technical:2}\ref{claim:lem:support-technical:2b}.
\item If $\mu\not\subseteq\mathcal{E}_t$, then \eqref{eq:U^a_t-support-function-formula} immediately gives claim~\ref{claim:lem:support-technical:2}\ref{claim:lem:support-technical:2d}.
\end{itemize}
\end{itemize}
In summary, we have shown that $\dom Z^a_t$ is compactly $i$-generated whenever
\[\mu\subseteq [\mathcal{E}^\ast_{t+1}\cap\mathcal{E}_t] \cup [\mathcal{E}^\ast_{t+1}\setminus\mathcal{E}_t] \cup [\mathcal{E}_t\setminus\mathcal{E}^\ast_{t+1}] = \mathcal{E}^\ast_t.\]
This concludes the inductive step, and completes the proof of Lemma \ref{lem:support-technical}.
\end{proof}

\end{document}